\documentclass[11pt]{article}
\usepackage[utf8]{inputenc}
\usepackage[backend=bibtex, maxbibnames=10]{biblatex}
\usepackage{subfig}
\usepackage[normalem]{ulem}
\usepackage{float}
\usepackage{times}

\bibliography{cc}

\usepackage{comment}

\usepackage{mathtools}

\usepackage{amsmath,amsthm,amssymb,latexsym,amsfonts, dsfont}
\usepackage{graphicx} 
\graphicspath{ {\string~/Desktop/} }
\usepackage[font=footnotesize,labelfont=bf]{caption}
\usepackage{mathabx}
\usepackage{hyperref}
\usepackage{lineno}
\usepackage{amsthm}
\usepackage{bm}
\usepackage{listings}
\usepackage[nottoc,numbib]{tocbibind}

\newtheorem{theorem}{Theorem}
\newtheorem{algorithm}{Algorithm}

\newtheorem{claim}{Claim}

\newtheorem{corollary}{Corollary}

\newtheorem{definition}{Definition}

\newtheorem{lemma}{Lemma}

\newtheorem{proposition}{Proposition}

\newtheorem{fact}{Fact}

\newcommand{\dn}{correlation metric}

\allowdisplaybreaks

\DeclareCiteCommand{\citeyearpar}
    {}
    {\mkbibparens{\bibhyperref{\printdate}}}
    {\multicitedelim}
    {}

\title{Simultaneously Approximating All $\ell_p$-norms\\ in Correlation Clustering }
\date{}

\author{Sami Davies\thanks{Department of EECS and the Simons Institute for the Theory of Computing, UC Berkeley. Supported by an NSF Computing Innovation Fellowship while at Northwestern University.} \and Benjamin Moseley\thanks{Carnegie Mellon University. Benjamin Moseley and Heather Newman were supported in part by  a Google Research Award, an Inform Research Award, a Carnegie Bosch Junior Faculty Chair, and NSF grants CCF-2121744  and  CCF-1845146.} \and Heather Newman$^\dagger$}

\begin{document}
\maketitle

\abstract{This paper considers correlation clustering on unweighted complete graphs. We give a combinatorial algorithm that returns a \emph{single} clustering solution that is \emph{simultaneously} $O(1)$-approximate for all $\ell_p$-norms of the disagreement vector;
in other words, a combinatorial $O(1)$-approximation of the \emph{all-norms} objective for correlation clustering.
This is the first proof that minimal sacrifice is needed in order to optimize different norms of the disagreement vector. In addition, our algorithm is the first combinatorial approximation algorithm for the $\ell_2$-norm objective, and more generally the first combinatorial algorithm for the $\ell_p$-norm objective when $1 < p < \infty$. 
It is also faster than all previous algorithms that minimize the $\ell_p$-norm of the disagreement vector, with run-time $O(n^\omega)$, where $O(n^\omega)$ is the time for matrix multiplication on $n \times n$ matrices. 
When the maximum positive degree in the graph is at most $\Delta$, this can be improved to a run-time of $O(n\Delta^2 \log n)$. 
}

\section{Introduction}
Correlation clustering is one of the most prominent problems in clustering, 
as it cleanly models community detection problems~\cite{veldt2018, SDELM21} and provides a way to decompose complex network structures~\cite{Wirth17, mccallum2004conditional}. 
The input to the unweighted correlation clustering problem is a complete graph $G=(V,E)$, where $|V|=n$ and each edge $e \in E$ is labeled positive $(+)$ or negative $(-)$.  
If the edge $(u,v)$ is positive, this indicates that $u$ and $v$ are similar, 
and analogously if the edge $(u,v)$ is negative, this indicates that $u$ and $v$ are dissimilar. 
The output of the problem is a partition of the vertex set into parts $C_1, C_2, \ldots$, where each part represents a cluster.   

The output should cluster similar vertices together and separate dissimilar vertices.
Specifically, for a fixed clustering (i.e., partition of the vertices), a positive edge $(u, v)$ is a \emph{disagreement} with respect to the clustering if $u$ and $v$ are in different clusters and an \emph{agreement} if $u$ and $v$ are in the same cluster.
Similarly, a negative edge $(u,v)$ is a disagreement with respect to the clustering if $u$ and $v$ are in the same cluster 
and an agreement if $u$ and $v$ are in different clusters.  
The goal is to find a clustering that minimizes some objective
that is a function of the disagreements.\footnote{Note that the sizes and number of clusters are unspecified.}
For example, the most commonly studied objective minimizes the total number of disagreements.

As an easy example to illustrate the problem,
consider a social network. Every pair of people has an edge between them, 
and the edge is positive if the two people have ever met before, and negative otherwise.
The goal of correlation clustering translates to partitioning all the people into clusters so that people are in the same cluster as their friends/acquaintances and in different clusters than strangers. 
The difficulty in constructing a clustering is that the labels may not be consistent, making disagreements unavoidable. Consider in the social network what happens when there is one person with two friends who have never met each other ($u,v,w$ with $(u,v)$ and $(u,w)$ positive but $(v,w)$ negative).
The choice of objective matters in determining the best clustering.

For a given clustering $\mathcal{C}$, let $y_{\mathcal{C}}(u)$ denote the number of edges incident to $u$ that are disagreements with respect to $\mathcal{C}$ (we drop $\mathcal{C}$ and write $y$ when it is clear from context).  
The most commonly considered objectives are ${\lVert y_{\mathcal{C}} \rVert}_p = \sqrt[p]{\sum_{u \in V} y_{\mathcal{C}}(u)^p}$ for $p \in \mathbb{R}_{\geq 1} \cup \{\infty\}$,
the $\ell_p$-norms of the \emph{disagreement vector} $y$.   
Note that the optimal objective values may drastically vary for different norms too.  (For instance, in the example in Appendix A of \cite{PM16}, $V=A \sqcup B$\footnote{$\sqcup$ denotes disjoint union.}, where $|A| = |B| = n/2$, and all edges are positive except for a negative matching between $A$ and $B$. The optimal $\ell_\infty$-norm objective value is 1 whereas the optimal for $\ell_1$ is $\Theta(n)$.)
When $p=1$, this objective minimizes the total number of disagreements.  Setting $p = \infty$ minimizes the maximum number of disagreements incident to any node, ensuring a type of worst-case fairness.\footnote{In the social network example, minimizing the $\ell_1$-norm corresponds to finding a clustering that minimizes the total number of friends who are separated plus the total number of strangers who are in the same cluster. The $\ell_\infty$-norm corresponds to finding a clustering minimizing the number of friends any person is separated from plus the number of strangers in that person's same cluster.}  
Balancing these two extremes---average welfare on one hand and fairness on the other---is the $\ell_2$-norm, which minimizes the variance of the disagreements at each node.

Correlation clustering was proposed by Bansal, Blum, and Chawla \citeyearpar{BBC04} with the objective of minimizing the $\ell_1$-norm of the disagreement vector.
The problem is \textsf{NP}-hard and several approximation algorithms have been proposed \cite{BBC04, ACN-pivot, chawla2015near, cohen2022correlation}.  
Puleo and Milenkovic \citeyearpar{PM16} proposed studying $\ell_p$-norms of the disagreement vector for $p >1$, and they give a $48$-approximation for any fixed $p$. 
Charikar, Gupta, and Schwartz \citeyearpar{CGS17} introduced an improved $7$-approximation, which  Kalhan,  Makarychev, and  Zhou \citeyearpar{KMZ19} further improved to a $5$-approximation. When $p >1$, up until recently, the only strategies were LP or SDP rounding, and it has been of interest to develop fast combinatorial algorithms \cite{veldt2022correlation}.   
Davies, Moseley, and Newman \citeyearpar{DMN23} introduced a combinatorial $O(1)$-approximation algorithm for $p= \infty$ (see also \cite{minmax4approx} for a different combinatorial algorithm),
and leave open the question of discovering a combinatorial $O(1)$-approximation algorithm for $1 < p < \infty$.

In all prior work, solutions obtained for $\ell_p$-norms are tailored to each norm (i.e., $p$ is part of the input to the algorithm), and 
it was not well-understood what the trade-offs were between solutions
that optimize different norms.  
Solutions naively optimizing one norm can be arbitrarily bad for other norms (see Figure \ref{fig: star}). 
A natural question is whether this loss from using a solution to one objective for another is avoidable. 
More specifically:

\begin{figure}
    \centering
\includegraphics[width = 8cm]{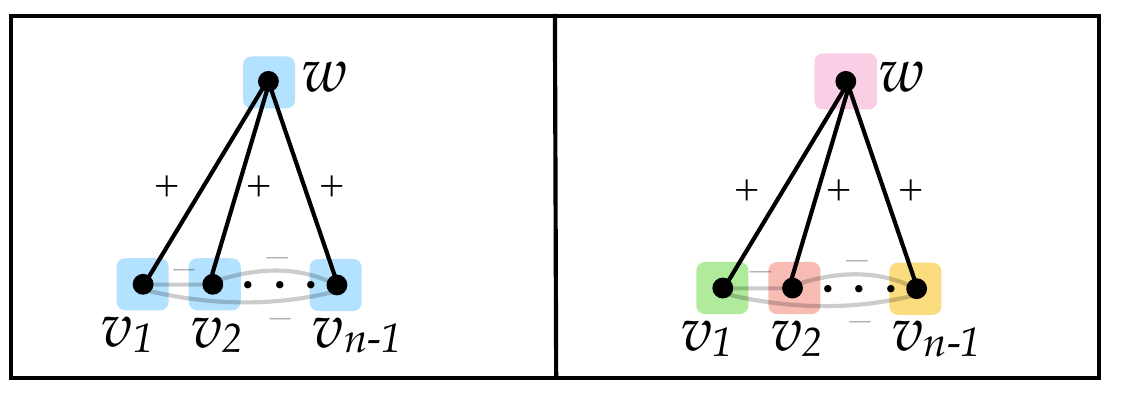}
    \caption{Two clusterings of the star graph, which has one node ($w$) with positive edges to all nodes, and the rest of the edges negative. \textbf{Left: }Clustering assigns all nodes to one (blue) cluster, and is (almost) optimal for the $\ell_\infty$-norm with cost $\Theta(n)$. 
    \textbf{Right: }Clustering assigns all nodes to different clusters and is (almost) optimal for the $\ell_1$-norm with cost $\Theta(n)$. The left solution is terrible for the $\ell_1$-norm, as the negative clique has $\Theta(n^2)$ edges that are disagreements. }
    \label{fig: star}
\end{figure}

\begin{center}
\emph{For any graph input to unweighted, complete correlation clustering,
does there exist a partition (clustering) that is \textbf{simultaneously} $O(1)$-approximate for all $\ell_p$-norm objectives?}
\end{center}

 Phrased another way, does there exist a \emph{universal algorithm} for $\ell_p$-norm correlation clustering---one which is guaranteed to produce a solution that well-approximates many objectives at once?
When the goal is to simultaneously optimize every $\ell_p$-norm, this is known as the \emph{all-norms} objective\footnote{In some of the literature, for instance that of Golovin et al. \cite{golovin2008all}, it is called the \emph{all-$\ell_p$-norms} objective.}. Universal algorithms and the all-norms objective are well-studied in combinatorial optimization problems, such as load balancing and set cover (see Section \ref{sec:related} for more discussion).  In the context of correlation clustering, such an algorithm outputs a partition that has good global performance (i.e. $\ell_1$-norm) and also has no individual node with too many adjacent disagreements (i.e. $\ell_\infty$-norm). Universal algorithms exist for some problems and are provably impossible for others. The question looms, what can be said about universal algorithms for correlation clustering?

As far as we are aware, there are no known results for the all-norms objective in other clustering problems. In fact, for the popular $k$-median and $k$-center problems, it is actually \textit{impossible} to $O(1)$-approximate (or even $o(\sqrt{n})$-approximate) these two objectives simultaneously \cite{2018bicriteria}.

\subsection{Results}
This paper is focused on optimizing all $\ell_p$-norms ($p \geq 1$) for correlation clustering at the same time.  The main result of the paper answers the previous question positively: perhaps surprisingly, there is a single clustering that simultaneously $O(1)$-approximates the optimal for all $\ell_p$-norms. Further, it can be found through an \emph{efficient combinatorial algorithm}.   This is also the first known combinatorial approximation algorithm for the $\ell_2$-norm objective and more generally $\ell_p$-norm objective for fixed $2 \leq p < \infty$.  

In what follows, let $O(n^\omega)$ denote the run-time of $n \times n$ matrix multiplication.  

\begin{theorem}\label{thm: main-acm}
Let $G=(V,E)$ be an instance of unweighted, complete correlation clustering on $|V|=n$ nodes. There exists a combinatorial algorithm returning a single clustering that is simultaneously an $O(1)$-approximation\footnote{Note this is independent of $p$.} for all $\ell_p$-norm objectives, for all $p\in \mathbb{R}_{\geq 1} \cup \{\infty\}$, and its run-time is $O(n^\omega)$.
\end{theorem}

The algorithm gives the \emph{fastest run-time} of any $O(1)$-approximation algorithm for the $\ell_p$-norm objective when $p  \in \mathbb{R}_{>1}$.
Further, the run-time can be improved  when the positive degree of the graph is bounded, as shown in the following corollary. 

\begin{corollary}\label{cor: main-acm-sparse}
Let $\Delta$ denote the maximum positive degree in an instance $G=(V,E)$ of unweighted, complete correlation clustering on $|V|=n$ nodes. Suppose $G$ is given as an adjacency list representation of its positive edges. There exists a combinatorial algorithm returning a single clustering that is simultaneously an $O(1)$-approximation for all $\ell_p$-norm objectives, for all $p\in \mathbb{R}_{\geq 1} \cup \{\infty\}$, and its run-time is $O(n\Delta^2 \log n)$.
\end{corollary}

The run-time of the algorithm matches the fastest known algorithm for the $\ell_\infty$-norm objective \cite{DMN23}, 
in both the general case and when the maximum positive degree is bounded.
The best-known algorithm before our work relied on solving a convex relaxation on $|V|^2$ variables and $|V|^3$ constraints.
We improve the run-time by avoiding this bottleneck.

In the setting when the positive edges form a regular graph, 
the interested reader may also find a rather beautiful proof (which is much simpler than that of Theorem \ref{thm: main-acm}) in Section \ref{sec: regular_warmup} showing there is a solution that is simultaneously $O(1)$-approximate for the $\ell_1$-norm and $\ell_{\infty}$-norm objectives.

\subsection{Related work}
\label{sec:related}
Correlation clustering was introduced by 
Bansal, Blum, and Chawla \citeyearpar{BBC04}.
The version they introduced also studies the problem on unweighted, complete graphs, but is concerned with minimizing the $\ell_1$-norm of the disagreement vector.
For this problem, Ailon, Charikar, and Newman \citeyearpar{ACN-pivot} designed the Pivot algorithm, which is a randomized algorithm that in expectation obtains a 3-approximation.
While we know algorithms with better approximations for $\ell_1$ correlation clustering than Pivot~\cite{chawla2015near, cohen2022correlation}, 
the algorithm remains a baseline in correlation clustering due to its simplicity. (However, Pivot can perform arbitrarily badly---i.e., give $\Omega(n)$ approximation ratios---for other $\ell_p$-norms; see again the example in Appendix A of \cite{PM16}.)
It is an active area of research to develop algorithms for the $\ell_1$-norm that focus on practical scalability~\cite{bonchi2014, chierichetti2014, pan2014scaling, SDELM21, single-pass2023}. 
Correlation clustering has also been studied on non-complete, weighted graphs \cite{CGS17, KMZ19}, with conditions on the cluster sizes \cite{puleo2015correlation}, and with asymmetric errors \cite{JafarovKMM21}. 
In fact, in recent work Veldt~\cite{veldt2022correlation} highlighted the need for deterministic techniques in correlation clustering that do not use linear programming.
Much interest in correlation clustering stems from its connections to applications,
including community detection, natural language processing, location area planning, and gene expression~\cite{veldt2018, SDELM21, Wirth17, mccallum2004conditional, ben1999clustering, demaine2003correlation}.

Puleo and Milenkovic \citeyearpar{PM16} introduced correlation clustering
with the goal of minimizing the $\ell_p$-norm of the disagreement vector.
They show that even for minimizing the $\ell_\infty$-norm on complete, unweighted graphs, the problem is \textsf{NP}-hard (Appendix C in \cite{PM16}).
Several groups found $O(1)$-approximation algorithms for minimizing the $\ell_p$-norm on complete, unweighted graphs \cite{PM16, CGS17, KMZ19}, the best of which is currently the 5-approximation of Kalhan, Makarychev, and Zhou \citeyearpar{KMZ19}.
Many other interesting objectives for correlation clustering focus on finding solutions that are (in some sense) fair or locally desirable~\cite{ahmadian2020fair, bateni2022scalable, ahmadi2020fair, friggstad2021fair, jafarov2021local, Khuller2019min}.
All of these previous works that study general $\ell_p$-norms or other notions of fairness or locality rely on solving a convex relaxation. 
This has two downsides: (1) the run-time of the algorithms are bottle-necked by the time it takes to solve the relaxation with at least $\Omega(n^2)$ many variables and $\Omega(n^3)$ constraints; in fact, it is time-consuming to even enumerate the $\Omega(n^2)$ variables and $\Omega(n^3)$ constraints; and (2) the solution is only guaranteed to be good for one particular value of $p$. 

 Several problems have been studied with the goal of finding a solution 
that is a good approximation for several objectives simultaneously.
The all-norms objective was introduced by Azar et al. \cite{azar2004all}, where the goal is to design a $\rho$-approximation algorithm for all $\ell_p$-norm objectives of a problem. 
They originally introduced the objective for the restricted assignment load balancing problem 
and showed an all-norms 2-approximation.
Further follow-up on the all-norms objective has been done for load balancing \cite{kleinberg1999fairness, bernstein2017simultaneously, langley2020improved}, and for set cover \cite{golovin2008all}.
The term ``universal" algorithm has also been used for Steiner tree \cite{busch2012split, busch2023one}, TSP \cite{jia2005universal}, and clustering \cite{ganesh2023universal}, 
though in these settings the goal is different, namely, to find a solution that is good for any potential \textit{input}; e.g., in Universal Steiner Tree, the goal is to find a spanning tree where for any set of terminals, the sub-tree connecting the root to the terminals is a good approximation of the optimal.

\section{Preliminaries}

We will introduce notation, 
and then we will discuss two relevant works---the papers by 
Kalhan, Makarychev, and Zhou \citeyearpar{KMZ19} and Davies, Moseley, and Newman \citeyearpar{DMN23}.

\subsection{Notation} \label{sec: notation}
Recall our input to the correlation clustering problem is 
$G=(V,E)$, an unweighted, complete graph on $n$ vertices, 
and every edge is assigned a label of either positive $(+)$ or negative $(-)$.
Let the set of positive edges be denoted $E^+$ and the set of negative edges $E^-$.
Then, we can define the \emph{positive neighborhood} and \emph{negative neighborhood} of a vertex $u$ as $N_u^+ =\{v\in V \mid (u,v) \in E^+\}$ and $N_u^- =\{v\in V \mid (u,v) \in E^-\}$, respectively.
We further assume without loss of generality that every vertex has a positive self-loop to itself.

A \emph{clustering} $\mathcal{C}$ is a partition of $V$ into \emph{clusters} $C_1,\ldots,C_k$ (but recall that $k$ is \textit{not} pre-specified). 
Let $C(u)$ denote the cluster that vertex $u$ is in, i.e., if $\mathcal{C}$ has $k$ clusters, there exists exactly one $i \in [k]$ such that $C(u) = C_i$.
It is also helpful to consider the vertices in a different cluster than $u$, and so we let $\overline{C(u)} = V \setminus C(u)$ denote this.
We say that a positive edge $e=(u,v) \in E^+$ is a \emph{disagreement} with respect to $\mathcal{C}$ if $v \in \overline{C(u)}$. On the other hand, we say that a negative edge $e=(u,v) \in E^-$ is a disagreement with respect to $\mathcal{C}$ if $v \in C(u)$.
For a fixed clustering $\mathcal{C}$, we denote the \emph{disagreement vector} of $\mathcal{C}$ as $y_{\mathcal{C}} \in \mathbb{Z}_{\geq 0}^n$, where for $u \in V$, $y_{\mathcal{C}}(u)$ is the number of edges incident to $u$ that are disagreements with respect to $\mathcal{C}$. We omit the subscript throughout the proofs when a clustering is clear.

Throughout, we let $\textsf{OPT}$ be the optimal objective value, 
and the $\ell_p$-norm to which it corresponds will be clear from context.
The next fact follows from the definitions seen so far (recalling also the positive self-loops).

\begin{fact} \label{pairwise_intersections}
For any $u,v \in V$, possibly with $u=v$, $n = |N_u^+ \cap N_v^+| + |N_u^- \cap N_v^-| + |N_u^+ \cap N_v^-| + |N_u^- \cap N_v^+|.$
\end{fact}

\subsection{Summary of work by Kalhan, Makarychev, and Zhou}
The standard linear program relaxation for correlation clustering 
is stated in LP \ref{KMZ_LP}.\footnote{Technically this is a convex program as the objective is convex. For simplicity we will refer to it as an LP as the constraints are linear.}
In the \textit{integer} LP, the variable $x_{uv}$ indicates whether vertices $u$ and $v$ will be in the same cluster (0 for yes, 1 for no), and the disagreement vector is $y$; the optimal solution to the integer LP has value \textsf{OPT}, while the optimal solution to the relaxation gives a lower bound on $\textsf{OPT}$. 
Note the triangle inequality is enforced on all triples of vertices, inducing a semi-metric space on $V$.
Throughout this paper, as in \cite{DMN23}, we refer to the algorithm by Kalhan, Makarychev, and Zhou as the KMZ algorithm. 
The \emph{KMZ algorithm} has two phases: it solves LP \ref{KMZ_LP}, 
and then uses the \textit{KMZ rounding algorithm} (Algorithm \ref{KMZ-alg}) to obtain an integral assignment of vertices to clusters. 
See Appendix \ref{sec: lp_rounding_alg} for its formal statement.
At a high-level, the KMZ rounding algorithm is an iterative, ball-growing algorithm
that uses the semi-metric to guide its choices on forming clusters.
Their algorithm is a 5-approximation, and produces different clusterings for different $p$, since the optimal solution $x^*$ to LP \ref{KMZ_LP} depends on $p$.

\vspace{-1cm}
\begin{flushleft}
\begin{equation*}\label{KMZ_LP}
\end{equation*}
LP \ref{KMZ_LP}
\end{flushleft}
\vspace{-1cm}
\begin{align}
     &\min {\lVert y \rVert}_p \notag\\
    \textsf{s.t. } y_u &= \sum_{v \in N_u^+} x_{uv} + \sum_{v \in N_u^-} (1-x_{uv})  && \forall u \in V \notag\\
     x_{uv}&\leq x_{uw}+x_{vw}  &&\forall u,v,w \in V \notag\\
     0 &\leq x_{uv} \leq 1  &&\forall u,v \in V. \notag
\end{align}

\begin{definition} \label{def: frac_cost}
    Let $f$ be a semi-metric on $V$, i.e., taking $x=f$ gives a feasible solution to LP \ref{KMZ_LP}. The \emph{fractional cost of $f$ in the $\ell_p$-norm objective} is the value of LP \ref{KMZ_LP} that results from setting $x=f$. When $p$ is clear from context, we will simply call this the fractional cost of $f$. 
\end{definition}

\subsection{Summary of work by Davies, Moseley, and Newman}\label{sec: DMN-summary}

The main take-away from the work of Kalhan, Makarychev, and Zhou \citeyearpar{KMZ19} is that one only requires a semi-metric on the set of vertices,
whose cost is comparable to the cost of an optimal solution,
as input to the KMZ rounding algorithm (Algorithm \ref{KMZ-alg}). 
Thus, the insight of Davies, Moseley, and Newman \citeyearpar{DMN23} for the $\ell_{\infty}$-norm objective is that one can combinatorially construct such a semi-metric without solving an LP, and at small loss in the quality of the fractional solution.
They do this by introducing the \emph{correlation metric}.

\begin{definition}[\cite{DMN23}] \label{def: orig_corr_metric}
For all $u,v \in V$, the \textit{\dn}~defines the 
distance between $u$ and $v$ as
\begin{align}
 \label{eq: equiv_def}
d_{uv} &= 1-\frac{|N_u^+ \cap N_v^+|}{|N_u^+ \cup N_v^+|} \notag \\
&=  \frac{|N_u^+ \cap N_v^-| +|N_u^- \cap N_v^+| }{|N_u^+ \cap N_v^+| + |N_u^+ \cap N_v^-| + |N_u^- \cap N_v^+|}.
\end{align}
\end{definition}

Note that the rewrite in Line (\ref{eq: equiv_def}) is apparent from Fact \ref{pairwise_intersections}.

The correlation metric captures useful information  succinctly.
Intuitively, if $u$ and $v$ have relatively large positive intersection, i.e., $N_u^+ \cap N_v^+$ is large compared to their other relevant joint neighborhoods $\left (N_u^+ \cap N_v^-\right ) \cup \left ( N_u^- \cap N_v^+\right )$, then from the perspective of $u$ and $v$,
fewer disagreements are incurred by putting
$u$ and $v$ in the same cluster than by putting them in different clusters.
This is because if $u$ and $v$ are in the same cluster, then they have disagreements on edges $(u,w)$ and $(v,w)$ for $w \in \left (N_u^+ \cap N_v^- \right )\cup \left (N_u^- \cap N_v^+\right )$, but if they are in different clusters, then $u$ and $v$ have disagreements on edges $(u,w)$ and $(v,w)$ for  $w \in N_u^+ \cap N_v^+$. 
Note that the metric is not normalized by $n$ (the size of the joint neighborhood), but instead by $n-| N_u^- \cap N_v^- | = |N_u^+ \cap N_v^+| + |N_u^+ \cap N_v^-| + |N_u^- \cap N_v^+|$.
In addition to the above intuition on the normalization factor, we also observe that $w \in N_u^- \cap N_v^- $ do not necessarily force disagreements on $(u,w)$ and $(v,w)$, since $w$ can go in a different cluster than both $u$ and $v$ without penalty.
For more on intuition behind the correlation metric, see Section 2 in \cite{DMN23}.

Davies, Moseley, and Newman \citeyearpar{DMN23} prove that the correlation metric $d$ can be used as input to the KMZ rounding algorithm by showing that (1) $d$ satisfies the triangle inequality and (2) the fractional cost of $d$ in the $\ell_{\infty}$-norm (recall Definition \ref{def: frac_cost}) is no more than 8 times the value of the optimal integral solution ($\textsf{OPT}$). Since the KMZ rounding algorithm (Algorithm \ref{KMZ-alg}) loses a factor of at most 5, inputting $d$ to that algorithm returns a 40-approximation algorithm. 
A benefit of the correlation metric is that it can be computed in time $O(n^{\omega})$, and even faster when the subgraph on positive edges is sparse.

\subsection{Technical overview}
 It is not hard to see that the correlation metric cannot be used as input to the KMZ algorithm for $\ell_p$-norms other than $p = \infty$, as one cannot bound the fractional cost of the correlation metric against the optimal with only an $O(1)$-factor loss. To see why, consider the star again, as in Figure \ref{fig: star}. Here, for all $u,v \in \{v_1,\ldots,v_{n-1}\}$, $d_{uv} = 1- 1/(n-(n-3)) = 2/3$, but for the $\ell_1$-norm, we need the semi-metric to have the value $1-d_{uv}$ be close to 0, i.e. $O(1/n)$, for such $u,v$, in order for the fractional cost to be comparable to the value of $\textsf{OPT}$ for $p=1$.

There are several possible fixes one could try to make to the correlation metric. 
One idea is that since one can interpret the correlation metric as a coarse approximation of the probability the Pivot algorithm separates $u$ and $v$\footnote{Pivot operates as follows: Choose a random unclustered $u \in V$. Take $u$ and all its unclustered positive neighbors and let this be the newest cluster. Continue until all vertices in $V$ are clustered.}, 
one could try to adapt the correlation metric to more accurately approximate this probability.\footnote{If in LP \ref{KMZ_LP}, $x_{uv}$ is set exactly to the probability that $u$ and $v$ are separated by Pivot, then $x$ will be a feasible solution with cost at most $\textsf{3OPT}$. However, this probability seems difficult to express in closed form, even approximately.}
Another idea, inspired by an observation below, is that one could define a semi-metric for edges in $E^+$ and another semi-metric for edges in $E^-$, but then there is the difficulty of showing the triangle inequality holds when positive and negative edges are mixed together. 
Both of these ideas were, for us, unsuccessful.

Instead, the following two observations of how the correlation metric works with respect to the $\ell_1$-norm led us to an effective adaptation: 
\begin{enumerate}
\item One can bound the fractional cost, \textit{restricted to positive edges}, of the correlation metric in the $\ell_1$-norm  by an $O(1)$-factor times the optimal solution's cost (see Claim \ref{claim: l1-pos-cm} in Appendix \ref{sec:l1-norm}). Negative edges still pose a challenge.
\item If the subgraph of positive edges is regular, then we show one can bound the fractional cost of the correlation metric in the $\ell_1$-norm on negative edges as well as positive.\footnote{In contrast, one cannot bound the fractional cost of the correlation metric on the (irregular) star (Figure \ref{fig: star}).} See Section \ref{sec: regular_warmup}. 
\end{enumerate}

These observations led us to ask whether some adjustments to the correlation metric might yield a semi-metric with bounded fractional cost in the $\ell_1$-norm or even the $\ell_p$-norm more generally (while still remaining bounded in the $\ell_{\infty}$-norm). Moreover, since the KMZ rounding algorithm does not depend on $p$ (whereas in the KMZ algorithm, the solution to the LP \textit{does} depend on $p$), inputting the same semi-metric to the rounding algorithm produces the same clustering for all $\ell_p$-norms! 

We are ready to define the \emph{adjusted correlation metric}. 
Let $\Delta_u$ denote the positive degree of $u$ (the degree of $u$ in the subgraph of positive edges).
\begin{definition} \label{def: adj_corr_metric}
    Define the \emph{adjusted correlation metric} $f: E \rightarrow [0,1]$ as follows:
    \begin{enumerate}
        \item For $d$ the correlation metric, i.e., 
        $d_{uv} = 1 - \frac{|N_u^+ \cap N_v^+|}{|N_u^+ \cup N_v^+|}$, initially set $f = d$. 
        \item If $e \in E^-$ and $d_e > 0.7$, set $f_e = 1$ (round up). 
        \item For $u \in V$ such that $|N_u^- \cap \{v: d_{uv} \leq 0.7\}| \geq \frac{10}{3}\Delta_u $, set $f_{uv}=1$ for all $v \in V \setminus \{u\}$.
    \end{enumerate}
\end{definition}

The idea in Step 3 is that if the fractional cost of negative edges incident to $u$ is sufficiently large, we instead trade this for the cost of positive disagreements, as 
the rounding algorithm will now put $u$ in its own cluster.
For the $\ell_\infty$-norm, this trade-off is innocuous.
For $\ell_p$-norms in general, a refined charging argument is needed to show that post-processing $d$ in this way sufficiently curbs the (too large) fractional cost of $d$.

In Section \ref{sec: regular_warmup}, we start with a warm-up exercise 
and show that if the graph on positive edges is regular, then 
the (original) correlation metric $d$ has $O(1)$-approximate fractional cost.
Note this section is not necessary to understanding the rest of the paper, 
but we include it in the main body because we find the proof here very clever!
The main technical result of the paper is Section \ref{sec: proofs}, where we prove Theorem \ref{thm: main-acm} 
by showing that the adjusted correlation metric can be input to the KMZ rounding algorithm.
Namely, we will first show (quite easily) that the adjusted correlation metric satisfies an approximate triangle inequality.
Then, it remains to upper bound the fractional cost of the adjusted correlation metric against $\textsf{OPT}$.
We tackle this with a combinatorial charging argument.   This argument leverages a somewhat different approach from that used in \cite{DMN23} and is simpler than their proof for only the $\ell_\infty$-norm.
The \textit{constant} approximation factor obtained from inputting the adjusted correlation metric to Algorithm \ref{KMZ-alg} is bounded above (and below) by universal constants for all $p$ (this is the worst case and one can get better constants for each $p$).
While we keep the argument for general $\ell_p$-norms in the main body of the paper,
the interested reader may find a simplified proof for the $\ell_1$-norm objective in Appendix \ref{sec:l1-norm}.

\section{A Special Case: Regular Graphs} \label{sec: regular_warmup}
In general, the original correlation metric $d$ (Definition \ref{def: orig_corr_metric}) does not necessarily have bounded fractional cost for the $\ell_1$-norm objective (or more generally for $\ell_p$-norm objectives). So, we use the adjusted correlation metric $f$ (Definition \ref{def: adj_corr_metric}) as input to the KMZ rounding algorithm (Appendix \ref{sec: lp_rounding_alg}). In this section, we show that \textit{if the subgraph of positive edges is regular}, then the correlation metric $d$ can be used as is (i.e., without the adjustments in Steps 2 and 3 of Definition \ref{def: adj_corr_metric}) to yield a clustering that is constant approximate for the $\ell_1$-norm and $\ell_{\infty}$-norm simultaneously:

\begin{theorem}
    Let $G=(V,E)$ be an instance of unweighted, complete correlation clustering, and let $E^+$ denote the set of positive edges. Suppose that the subgraph induced by $E^+$ is regular. Then the fractional cost of $d$ in the $\ell_1$-norm objective is within a constant factor of $\textsf{OPT}$:

    \[\sum_{u \in V} \sum_{v \in N_u^+} d_{uv} + \sum_{u \in V} \sum_{v \in N_u^-} (1-d_{uv}) = O(\textsf{OPT}).\]
    Therefore, the clustering produced by inputting $d$ to the KMZ rounding algorithm is a constant-factor approximation simultaneously for the $\ell_1$-norm and $\ell_{\infty}$-norm objectives. 
\end{theorem}

While this proof is quite different than the one for arbitrary-degree graphs,
we include it because we find it beautiful,
and it gives useful intuition for where the correlation metric breaks down for irregular graphs.

\begin{proof}
    Let $\Delta$ be the (common) degree of the positive subgraph. To show that the fractional cost of $d$ in the $\ell_1$-norm objective is $O(\textsf{OPT})$ for regular graphs, we will use a dual fitting argument. The LP relaxation we consider is from \cite{ACN-pivot}, which uses a dual fitting argument to show constant approximation guarantees for Pivot (although the proof here does not otherwise resemble the proof for Pivot). The primal is given by 

\begin{equation} \label{eq: lp_pivot} \tag{$P$}
 \min\Big\{\sum_{e \in E} x_e \mid x_{ij}+x_{jk}+x_{ki} \geq 1, \forall ijk \in \mathcal{T}, x \geq 0\Big\}
\end{equation}
where $\mathcal{T}$ is the set of bad triangles (i.e. triangles with exactly two positive edges and one negative edge). For $x \in \{0,1\}^{|E|}$, $x$ corresponds to disagreements in a clustering: we set $x_e = 1$ if $e$ is a disagreement and $x_e = 0$ otherwise. The constraints state that every clustering must make a disagreement on every bad triangle. Thus, ($P$) is a relaxation for the $\ell_1$-norm objective. In fact, we will prove the stronger statement that the fractional cost is $O(\textsf{OPT}_{P})$, where $\textsf{OPT}_{P}$ is the optimal objective value of ($P$). 

The dual is given by 
\begin{equation} \label{eq: dual_pivot} \tag{$D$}
\max\Big\{\sum_{T \in \mathcal{T}}y_T \mid \sum_{T \in \mathcal{T}: T \ni e} y_T \leq 1, \forall e \in E, y \geq 0 \Big\}.
\end{equation}

We show that by setting $y_T = \frac{1}{2\Delta}$ for all $T \in \mathcal{T}$, $y$ satisfies the following properties:
\begin{enumerate}
    \item $y$ is feasible in ($D$).
    \item The fractional cost of $d$ is at most $6 \cdot \sum_{T \in \mathcal{T}} y_T$.
\end{enumerate}

Letting $\textsf{OPT}_D$ be the optimal objective value of ($D$), we have $6 \cdot \sum_{T \in \mathcal{T}} y_T \leq 6 \cdot \textsf{OPT}_D = 6 \cdot \textsf{OPT}_{P} \leq 6 \cdot \textsf{OPT}$, which will conclude the proof.

To prove feasibility, we case on whether $e$ is positive or negative. 

If $e \in E^-$, then $|\{T \in \mathcal{T}: T \ni e\}| = |N_u^+ \cap N_v^+| \leq \Delta$, where equality is by the definition of a bad triangle. So 
\[\sum_{T \in \mathcal{T}: T \ni e} y_T  \leq \frac{\Delta}{2\Delta} \leq 1.\]

If $e \in E^+$, then $|\{T : T \ni e\}| = |(N_u^+ \cap N_v^-) \cup (N_u^- \cap N_v^+)| \leq 2\Delta$. We conclude $y$ is feasible, since
\[\sum_{T \in \mathcal{T}: T \ni e} y_T  \leq \frac{2\Delta}{2\Delta} = 1.\]

Now we need to show that the fractional cost of $d$ is bounded in terms of the objective value of $(D)$. First we bound the fractional cost of the negative edges:
\[\sum_{(u,v) \in E^-}(1-d_{uv}) \leq \sum_{(u,v)\in E^-} |N_u^+ \cap N_v^+|/\Delta= \sum_{e \in E^-} \sum_{T \in \mathcal{T}: T \ni e} 1/\Delta = \sum_{e \in E^-} \sum_{T \in \mathcal{T}: T \ni e} 2y_T,\]
where in the first inequality we have used that $|N_u^+ \cup N_v^+| \geq \Delta$. Next we bound the fractional cost of the positive edges:
\[\sum_{(u,v) \in E^+}d_{uv} \leq \sum_{(u,v)\in E^+} (|N_u^+ \cap N_v^-| + |N_u^- \cap N_v^+|)/\Delta = \sum_{e \in E^+} \sum_{T \in \mathcal{T}: T \ni e} 1/\Delta = \sum_{e \in E^+} \sum_{T \in \mathcal{T}: T \ni e} 2y_T. \]

So the total fractional cost is bounded by 
$\sum_{e \in E} \sum_{T \in \mathcal{T}: T \ni e} 2 y_T = 6 \cdot \sum_{T \in \mathcal{T}} y_T,$
since each triangle contains three edges. This is what we sought to show. Since the fractional cost of $d$ is bounded for the $\ell_1$-norm objective (and the $\ell_{\infty}$-norm objective by \cite{DMN23}), using $d$ as input to KMZ rounding algorithm produces a clustering that is simultaneously $O(1)$-approximate for the $\ell_1$- and $\ell_{\infty}$-norm objectives.  
\end{proof}

\section{Proof of Theorem \ref{thm: main-acm}}\label{sec: proofs}

The goal of this section is to prove Theorem \ref{thm: main-acm} and the subsequent Corollary 
\ref{cor: main-acm-sparse}.
We begin by outlining that the adjusted correlation metric satisfies an approximate triangle inequality in Subsection \ref{sec: tri-ineq}. 
Then in Subsection \ref{sec: lp-norm}, we prove the fractional cost of the adjusted correlation metric in any $\ell_p$-norm objective is an $O(1)$ factor away from the optimal solution's value.
We tie it all together to prove Theorem \ref{thm: main-acm} and Corollary 
\ref{cor: main-acm-sparse} in Subsection \ref{sec: thm1_proof}.

Before all that, we prove a proposition that will be key in several settings. Loosely, it states that if two vertices are close to each other according to $d$, then they have a large shared positive neighborhood. 
 \begin{proposition} \label{prop: pos_overlap}
     Fix vertices $u,v \in V$ and a clustering $\mathcal{C}$ on $V$ such that $d_{uv} \leq 0.7$ and $|N_u^+ \cap C(u)|\hspace{1mm}/ \hspace{1mm} |N_u^+|\geq 0.85$.
     Then $|N_u^+ \cap N_v^+ \cap C(u)| \geq 0.15 \cdot |N_u^+|$.
 \end{proposition}
\begin{proof}
    Since $d_{uv} \leq 0.7$, 
    $\frac{|N_u^+ \cap N_v^+|}{|N_u^+|} \geq \frac{|N_u^+ \cap N_v^+|}{|N_u^+ \cup N_v^+|} \geq 0.3.$
    Using the assumption on $u$ together with the above inequality implies that $|N_v^+ \cap N_u^+ \cap C(u)|\hspace{1mm}/ \hspace{1mm}|N_u^+| \geq 0.85 + 0.3 - 1 = 0.15.$
\end{proof}

\subsection{Triangle inequality}\label{sec: tri-ineq}
Recall that the correlation metric $d$ satisfies the triangle inequality (see Section 4.2 in \cite{DMN23}). We will show that the adjusted correlation metric $f$ satisfies an approximate triangle inequality, which is sufficient for the KMZ rounding algorithm. 
Formally, we say that a function $g$ is a \emph{$\delta$-semi-metric} on some set $S$ if it is a semi-metric on $S$, except instead of satisfying the triangle inequality, $g$ satisfies 
$g(u,v) \leq \delta \cdot( g(u,w) + g(v,w))$ for all $u,v,w \in S$.

\begin{lemma}[Triangle Inequality] \label{lem: adjusted_tri_ineq}
    The adjusted correlation metric $f$ is a $\frac{10}{7}$-semi-metric. 
\end{lemma}

The proof of Lemma \ref{lem: adjusted_tri_ineq} is straightforward given that $d$ satisfies the triangle inequality; see Appendix \ref{sec: tri_ineq}.

Lemma 3 in \cite{DMN23} proves that one can input a semi-metric that satisfies an approximate triangle inequality (instead of the exact triangle inequality) to the KMZ rounding algorithm (with some loss in the approximation factor). We summarize the main take-away below.

\begin{lemma}[\cite{DMN23}]\label{lem: apx-tri}
    If $g$ is a $\delta$-semi-metric on the set $V$, instead of a true  semi-metric (i.e., $1$-semi-metric), then the KMZ algorithm loses a factor of $1+\delta+\delta^2+\delta^3+\delta^4$.\footnote{When $\delta =1$, this factor equals 5, which is the loss in the KMZ algorithm.}
\end{lemma}
Since we show in Lemma \ref{lem: adjusted_tri_ineq} that $f$ is a $\frac{10}{7}$-semi-metric,
we lose a factor of 12 in inputting $f$ to the KMZ algorithm (along with the factor loss from the fractional cost).

\subsection{Bounding the fractional cost of $\ell_p$-norms}\label{sec: lp-norm}
This section bounds the fractional cost of the adjusted correlation metric for the $\ell_p$-norms.  The following lemma considers the case where $p=\infty$.  The general case is handled in the subsequent lemma.

\begin{lemma} \label{lem: infty_proof}
    The fractional cost of the adjusted correlation metric $f$ in the $\ell_{\infty}$-norm objective is at most $56 \cdot \textsf{OPT}$, where $\textsf{OPT}$ is the cost of the optimal integral solution in the $\ell_{\infty}$-norm.
\end{lemma}

The lemma follows from the fact that the fractional cost of the correlation metric $d$ in the $\ell_{\infty}$-norm is known to be bounded by \cite{DMN23}, and that it only decreases when $d$ is replaced by $f$ due to Definition \ref{def: adj_corr_metric}. For completeness, we include a proof in Appendix \ref{sec: infty_norm}. 

For finite $p \in \mathbb{R}_{\geq 1}$, our techniques hold for any $\ell_p$-norm, as we only require the convexity of $x^p$. As far as we know, there has been no study initiated on the problem for $0 \leq p<1$.
We use two primary lemmas---one for the positive edge fractional cost and one for the negative edge fractional cost---to show that the adjusted correlation metric well approximates the optimal for general $\ell_p$-norms.

\begin{lemma}\label{lem:lp-bdd}
    The fractional cost of the adjusted correlation metric $f$ in the $\ell_p$-norm objective is a constant factor (independent of $p$) away from the cost of the optimal integral solution in the $\ell_p$-norm.
\end{lemma}
\begin{proof}
Let $y$ be the disagreement vector for an optimal clustering $\mathcal{C}$ in the $\ell_p$-norm, for any $p \in \mathbb{R}_{\geq 1} \cup \{\infty\}$. 
When $p = \infty$, see Lemma \ref{lem: infty_proof}.
For $p\in \mathbb{R}_{\geq 1}$,
by definition $\textsf{OPT}^p = \sum_{w \in V} (y(w))^p,$
and the $p$th power of the fractional cost of $f$ is given by 
\[cost(f)^p = \sum_{u \in V} \Big[ \sum_{v \in N_u^+} f_{uv} + \sum_{v \in N_u^-} (1-f_{uv})\Big]^p.\]
\[\text{Observe that }\quad cost(f)^p \leq 2^p \underbrace{\sum_{u \in V} \Big(\sum_{v \in N_u^+} f_{uv} \Big)^p}_{(S^+)^p} + 2^p \underbrace{\sum_{u \in V} \Big(\sum_{v \in N_u^-} (1-f_{uv}) \Big)^p}_{(S^-)^p}.\]

We refer to bounding $(S^+)^p$ as bounding the fractional cost of the positive edges, and likewise $(S^-)^p$ for the negative edges. 
The first sum, $(S^+)^p$, is bounded in Lemma \ref{lem: lp_bdd-pos} and the second sum, $(S^-)^p$, is bounded in Lemma \ref{lem: lp_bdd-neg}. Using those two bounds, together we have 
$cost(f) \leq [2^p((S^+)^p + (S^-)^p)]^{1/p} \leq 529,$
for $p \in [1, \infty)$. Specifically, the middle term is maximized at $p=1$, giving the bound of 529, and tends to below 214 as $p \rightarrow \infty$. (A more tailored analysis in Appendix \ref{sec:l1-norm} gives a constant of 74 for $p=1$.)
\end{proof}
We note that, as our main interest is determining whether a simultaneous constant approximation is even possible (and a combinatorial one, at that), we did not pay particular attention to optimizing constants, but suspect these could be greatly reduced.

\subsubsection{Fractional cost of positive edges in $\ell_p$-norms} \label{sec: pos_edges_lp}

Bounding the fractional cost of the positive edges in the $\ell_p$-norm will be similar to in the simpler $\ell_1$-norm (see Appendix \ref{sec:l1-norm}), only there will be an extra step in which we apply Jensen's inequality. 
\begin{lemma} \label{lem: lp_bdd-pos}
    For $p \in \mathbb{R}_{\geq 1}$, the fractional cost of the adjusted correlation metric $f$ in the $\ell_p$-norm objective for the set of positive edges is a constant factor approximation to the optimal, i.e.,
    $$(S^+)^p =  \sum_{u \in V} \Big(\sum_{v \in N_u^+} f_{uv} \Big)^p \leq 
2^p \cdot [(8^p/2 +1)((20/3)^p + 2 + 2 \cdot 4^p) + 8^p +1 ] \cdot \textsf{OPT}^p.
$$
\end{lemma}

One of the challenges in bounding the cost of $f$ is that disagreements in the $\ell_p$-norm objective for $p \neq 1$ are asymmetric, in that a disagreeing edge charges $y(u)$ and $y(v)$ (whereas for $p=1$ we can just sum the number of disagreeing edges). Step 3 rounds up the edges incident to $u$ when the tradeoff is good \textit{from $u$'s perspective}. However, an edge $(u,v)$ may be rounded up to 1 when this tradeoff is good from $v$'s perspective, but not from $u$'s perspective. The high-level idea for why this is fine is that if $u$ and $v$ are close under $d$, their positive neighborhoods overlap significantly and, in some average sense, $u$ can charge to $v$. Proving this requires a double counting argument using a bipartite auxiliary graph. If $u$ and $v$ are far under $d$, on the other hand, we can charge to the cost of the correlation metric, which will be bounded on an appropriate subgraph.

The second challenge is showing that the $\ell_p$-norm of the disagreement vector, restricted to vertices $u$ that are made singletons in Step 3, is bounded. This again requires a double counting argument.

\begin{proof}
Fix an optimal clustering $\mathcal{C}$. 
We partition vertices based on membership in $C(u)$ or $\overline{C(u)}$ (as defined in Subsection \ref{sec: notation}). Let $y$ denote the disagreement vector of $\mathcal{C}$. We have
\begin{align*}
    (S^+)^p &= \sum_{u \in V} \Big( \sum_{v \in N_u^+} f_{uv}\Big)^p 
    \leq 2^p \underbrace{\sum_{u \in V} \Big(\sum_{v \in N_u^+ \cap C(u)} f_{uv}\Big)^p}_{S_1^+} + 2^p \underbrace{\sum_{u \in V} \Big(\sum_{v \in N_u^+ \cap \overline{C(u)}} f_{uv}\Big)^p}_{S_2^+}.
\end{align*}
 It is easy to bound $S_2^+$ by using the trivial upper bound $f_{uv} \leq 1$:
\[\boxed{S_2^+ = \sum_{u \in V} \Big(\sum_{v \in N_u^+ \cap \overline{C(u)}} f_{uv}\Big)^p \leq \sum_{u \in V} \Big(\sum_{v \in N_u^+ \cap \overline{C(u)}} 1 \Big)^p \leq \sum_{u \in V} (y(u))^p = \textsf{OPT}^p,}\]
where we have used that every edge $(u,v) \in E^+$ with $v \not \in C(u)$ is a disagreement incident to $u$.

Next, we bound $S_1^+$. Let $R_1$ be the set of $u$ for which Step 3 of Definition \ref{def: adj_corr_metric} applies. For these $u$, we have $f_{uv} = 1$ for all $v \in V \setminus \{u\}$. Let $R_2 = V \setminus R_1$. For $u \in R_2$ and $v \in N_u^+$, we have that either $v\in R_2$, in which case $f_{uv} = d_{uv}$; or $v \in R_1$, in which case $f_{uv}=1$. (Note that $V$ is the disjoint union of $R_1$ and $R_2$.) So 
\begin{align*}S_1^+ = \underbrace{\sum_{u \in R_1} \Big(\sum_{v \in N_u^+ \cap C(u), v \neq u} 1 \Big)^p}_{S_{11}^+} + \underbrace{\sum_{u \in R_2}\Big(\sum_{v \in N_u^+ \cap C(u)} f_{uv}\Big)^p}_{S_{12}^+}, \end{align*}
and in particular 
\begin{align*}
    S_{12}^+ &= \sum_{u \in R_2}\Big(\sum_{v \in N_u^+ \cap C(u) \cap R_1} 1 + \sum_{v \in N_u^+ \cap C(u) \cap R_2} d_{uv} \Big)^p \\
    &= \sum_{u \in R_2}\Big(\sum_{\substack{v \in N_u^+ \cap C(u) \cap R_1\\ d_{uv} \leq 1/4}} 1 + \sum_{\substack{v \in N_u^+ \cap C(u) \cap R_1\\ d_{uv} \geq 1/4}} 1 + \sum_{v \in N_u^+ \cap C(u) \cap R_2} d_{uv} \Big)^p \\
    &\leq \sum_{u \in R_2}\Big(\sum_{\substack{v \in N_u^+ \cap C(u) \cap R_1\\ d_{uv} \leq 1/4}} 1 + \sum_{\substack{v \in N_u^+ \cap C(u) \cap R_1\\ d_{uv} \geq 1/4}} 4 \cdot d_{uv} + \sum_{v \in N_u^+ \cap C(u) \cap R_2} d_{uv} \Big)^p \\
    &\leq \sum_{u \in R_2}\Big(\sum_{\substack{v \in N_u^+ \cap C(u) \cap R_1\\ d_{uv} \leq 1/4}} 1 +  \sum_{v \in N_u^+ \cap C(u)} 4 \cdot d_{uv} \Big)^p 
    \leq 2^p  \underbrace{\sum_{u \in R_2}\Big(\sum_{\substack{v \in N_u^+ \cap R_1\\ d_{uv} \leq 1/4}} 1 \Big)^p}_{S_{13}^+} + 8^p \cdot \underbrace{\sum_{u \in R_2} \Big(\sum_{v \in N_u^+ \cap C(u)}  d_{uv} \Big)^p}_{S_{14}^+}.
\end{align*}
First we will bound $S_{13}^+$. To do so, we will strongly use that $d_{uv} \leq 1/4$ in the inner sum. In particular, we will make use of the following easy proposition. 
\begin{proposition} \label{prop: similar_nbhds}
Let $d$ be the correlation metric, and suppose $d_{uv} \leq 1/4$. Then $|N_u^+| \leq \frac{7}{3} \cdot |N_v^+|$.
\end{proposition}
\begin{proof}[Proof of Proposition \ref{prop: similar_nbhds}]
    Since $d_{uv} \leq 1/4$ implies $1-d_{uv} \geq 3/4$, we see that
    \[\frac{|N_u^+ \cap N_v^+|}{|N_u^+| + |N_v^+| - |N_u^+ \cap N_v^+|} \geq 3/4,\]
    and we complete the proof since
    $7\cdot|N_v^+| \geq 7\cdot|N_u^+ \cap N_v^+| \geq 3\cdot|N_u^+| + 3\cdot|N_v^+| \geq 3\cdot|N_u^+|.$
\end{proof}
Next, we will need to create a bipartite auxiliary graph $H = (R_2, R_1, F)$ with $R_2$ and $R_1$ being the sides of the partition, and $F$ being the edge set. We will then use a double counting argument. Place an edge between $u \in R_2$ and $v \in R_1$ if $uv \in E^+$ and $d_{uv} \leq 1/4$. Then we have precisely that 
\[S_{13}^+ = \sum_{u \in R_2} \text{deg}_{H}(u)^p.\] 
 We will show that 
\begin{equation} \label{eq: aux_desired_bound}
  \boxed{S_{13}^+ = \sum_{u \in R_2} \text{deg}_H(u)^p \leq  4^{p-1} \cdot \sum_{v \in R_1} |N_v^+|^p \leq 4^{p-1} \cdot \left((20/3)^p + 2 + 2 \cdot 4^p \right) \cdot \textsf{OPT}^p}
\end{equation}
where the last bound  follows from Proposition \ref{prop: R1_bounding}, which we establish separately below. We will bound via double counting the quantity $L$, defined below. Let $N_H(\cdot)$ denote the neighborhoods in $H$ of the vertices. 
\begin{align}
    L &:= \sum_{f=uv \in F} \left(\text{deg}_H(u) + \text{deg}_H(v) \right)^{p-1}  
    \leq \sum_{v \in R_1} \sum_{u \in N_{H}(v)} \left(\text{deg}_H(v) + \text{deg}_H(u) \right)^{p-1} \notag \\
    &\leq \sum_{v \in R_1} \sum_{u \in N_{H}(v)} \left( |N_v^+| + |N_u^+| \right)^{p-1}  
    \leq \sum_{v \in R_1} \sum_{u \in N_{H}(v)} 4^{p-1} \cdot |N_v^+|^{p-1} \label{eq: closeness_prop} \\
    &\leq 4^{p-1} \cdot \sum_{v \in R_1} |N_v^+| \cdot |N_v^+|^{p-1} \notag 
    = 4^{p-1} \cdot \sum_{v \in R_1} |N_v^+|^p \notag
\end{align}
 where in (\ref{eq: closeness_prop}) we've used Proposition \ref{prop: similar_nbhds}. Note that $L$ is upper bounded by the right-hand side in (\ref{eq: aux_desired_bound}). Now it just remains to show that $L$ is lower bounded by the left-hand side in (\ref{eq: aux_desired_bound}).
\begin{align*}
    L &= \sum_{f=uv \in F} \left(\text{deg}_H(u) + \text{deg}_H(v) \right)^{p-1} 
    = \sum_{u \in R_2} \sum_{v \in N_H(u)} \left(\text{deg}_H(u) + \text{deg}_H(v) \right)^{p-1} \\
    &\geq \sum_{u \in R_2} \sum_{v \in N_H(u)} \text{deg}_H(u)^{p-1} 
    = \sum_{u \in R_2} \text{deg}_H(u) \cdot \text{deg}_H(u)^{p-1} 
    = \sum_{u \in R_2} \text{deg}_H(u)^{p},
\end{align*}
 which is what we sought to show.
Now we bound $S_{14}^+$. 
\begin{align*}
    S_{14}^+ &\leq \sum_{u \in V} \Bigg(\sum_{v \in N_u^+ \cap C(u)} \frac{|N_u^+ \cap N_v^-| + |N_u^- \cap N_v^+|}{|N_u^+ \cup N_v^+|} \Bigg)^p \\
    &\leq \sum_{u \in V} \Bigg(\sum_{v \in N_u^+} \frac{y(u) + y(v)}{|N_u^+ \cup N_v^+|} \Bigg)^p 
    \leq \sum_{u \in V} |N_u^+|^{p-1} \sum_{v \in N_u^+} \frac{(y(u) + y(v))^p}{|N_u^+ \cup N_v^+|^p} \\
    &\leq 2^p\sum_{u \in V}\sum_{v \in N_u^+} |N_u^+|^{p-1} \cdot \frac{y(u)^p}{|N_u^+ \cup N_v^+|^p} + 2^p\sum_{u \in V}\sum_{v \in N_u^+} |N_u^+|^{p-1} \cdot \frac{y(v)^p}{|N_u^+ \cup N_v^+|^p}. 
\end{align*}
In the second line, the first inequality uses the fact that for $w \in (N_u^+ \cap N_v^-) \cup (N_u^- \cap N_v^+)$, then at least one of $(u,w), (v,w)$ is a disagreement, since $v \in C(u)$ in the inner summation of the first line. The second inequality in the second line uses Jensen's inequality. 

To bound the first double sum above, we use an averaging argument:
$$\sum_{u \in V}\sum_{v \in N_u^+} |N_u^+|^{p-1} \cdot \frac{y(u)^p}{|N_u^+ \cup N_v^+|^p} \leq \sum_{u \in V} \sum_{v \in N_u^+} \frac{y(u)^p}{|N_u^+|} = \sum_{u \in V} y(u)^p = \textsf{OPT}^p.$$
To bound the second double sum, we first have to flip it:
\begin{align*}
\sum_{u \in V}\sum_{v \in N_u^+} |N_u^+|^{p-1} \cdot \frac{y(v)^p}{|N_u^+ \cup N_v^+|^p} &= \sum_{v \in V} \sum_{u \in N_v^+} |N_u^+|^{p-1} \cdot \frac{y(v)^p}{|N_u^+ \cup N_v^+|^p} \\
&\leq \sum_{v \in V} \sum_{u \in N_v^+} |N_u^+|^{p-1} \frac{y(v)^p}{|N_u^+|^{p-1} \cdot |N_v^+|} 
= \sum_{v \in V} y(v)^p = \textsf{OPT}^p.
\end{align*}
\text{In total, we have }
$$\boxed{S_{14}^+ \leq 2 \cdot 2^p \cdot \textsf{OPT}^p = 2^{p+1} \cdot \textsf{OPT}^p}$$ 

$$\text{and } \quad\boxed{S_{12}^+ \leq 2^p \cdot S_{13}^+ + 8^p \cdot S_{14}^+ \leq 2^p \cdot 4^{p-1} \cdot \left((20/3)^p + 2 + 2 \cdot 4^p \right) \cdot \textsf{OPT}^p + 8^p \cdot \textsf{OPT}^p.}$$

Next we turn to bounding $S_{11}^+$. Recall that $R_1 = \{u : |N_u^- \cap \{v: d_{uv} \leq 0.7\}| \geq \frac{10}{3} \cdot \Delta_u\} $ 
and
\[S_{11}^+ \leq \sum_{u \in R_1} |N_u^+ \cap C(u)|^p \leq \sum_{u \in R_1} |N_u^+|^p.\]

So it suffices to bound the right-hand side, which we do in the following proposition. 
\begin{proposition} \label{prop: R1_bounding} Let $R_1$ be the set of $u$ for which Step 3 of Definition \ref{def: adj_corr_metric} applies. Then 
    $$\sum_{u \in R_1} |N_u^+|^p \leq \left((20/3)^p + 2 + 2 \cdot 4^p \right) \cdot \textsf{OPT}^p.$$
\end{proposition}

\begin{proof}[Proof of Proposition \ref{prop: R1_bounding}]
For $u \in R_1$, define 
$R_1(u) = N_u^- \cap \{v: d_{uv} \leq 0.7\},$
so in particular $|R_1(u)| \geq \frac{10}{3} \cdot \Delta_u$. Fix a vertex $u \in R_1$. We consider a few cases. The crux of the argument is Case 2a(ii). 

\paragraph{Case 1: }\emph{At least a 0.15 fraction of $N_u^+$ is in clusters other than $C(u)$.} \\
Let $ u \in V^{1}$ be the vertices in this case.
This means that $0.15 \cdot |N_u^+| \leq y(u)$, so 
\[\boxed{\sum_{u \in V^1} |N_u^+|^p \leq \sum_{u \in V^1} \frac{1}{0.15^p} y(u)^p \leq (20/3)^p \cdot \textsf{OPT}^p.}\]

\paragraph{Case 2: }\emph{At least a 0.85 fraction of $N_u^+$ is in $C(u)$. }\\
We further partition the cases based on how much $R_1(u)$ intersects $C(u)$.\\

\begin{figure}
    \centering
\includegraphics[width = 8cm]{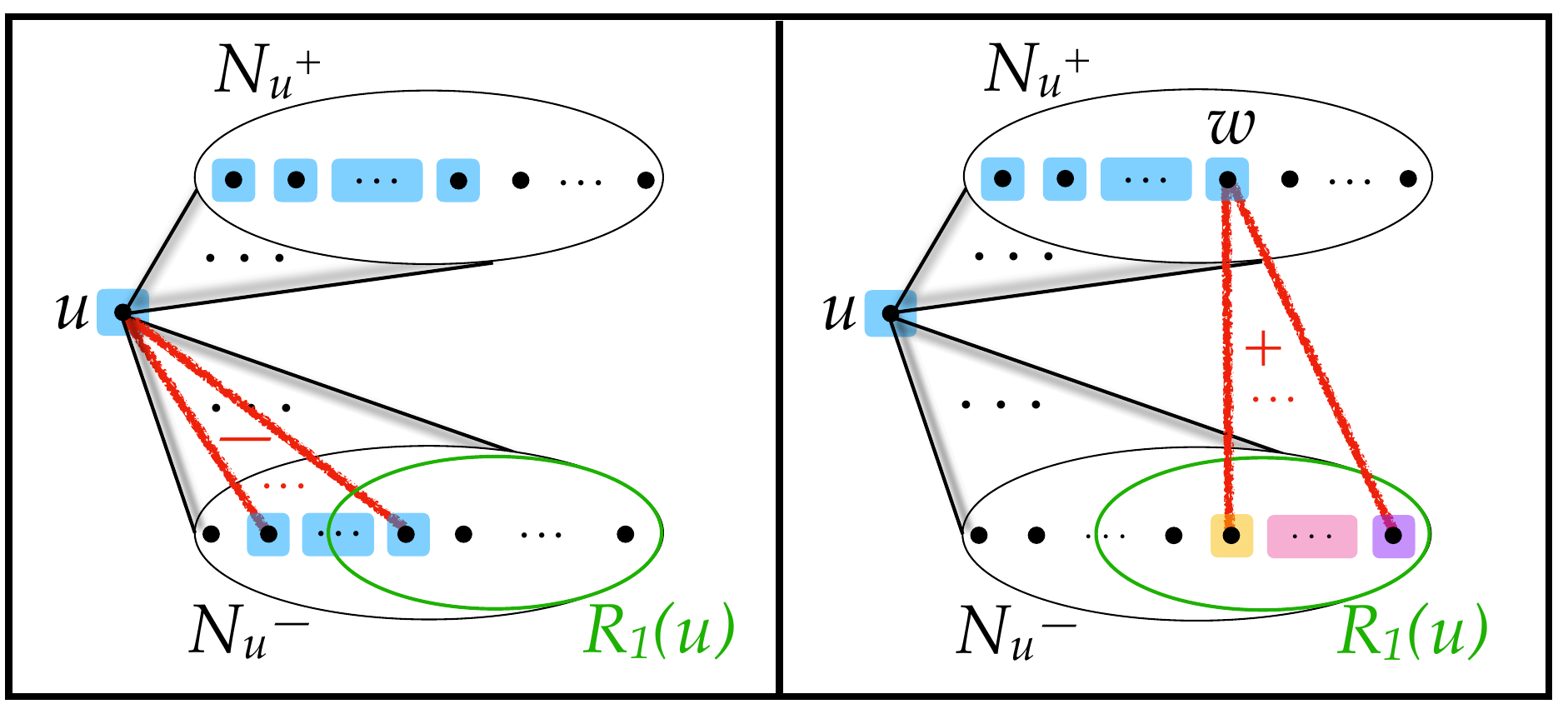}
    \caption{\textbf{Left:} Case 2a(i). For $v \in |N_u^- \cap C(u)|$, $(u,v)$ is a disagreement.
    \textbf{Right:} Case 2a(ii). For $w \in |N_u^+ \cap C(u)|$ and $v \in N_w^+ \cap R_1'(u)$, $(w,v)$ is a disagreement.}
    \label{fig: case2a}
\end{figure}

\noindent\textbf{Case 2a:} \emph{ At least half of $R_1(u)$ is in clusters other than $C(u)$.} \\
We partition into cases (just one more time!) based on the size of $N_u^- \cap C(u)$. See Figure \ref{fig: case2a}.
    \begin{itemize}
        \item \textbf{Case 2a(i):} \emph{At least half of $R_1(u)$ is in clusters other than $C(u)$ and $|N_u^- \cap C(u)| \geq \Delta_u$. }\\
        Let $ u \in  V^{2a(i)}$ be the vertices in this case.
        Note that $y(u) \geq |N_u^- \cap C(u)|$. Then
        \begin{align*}
        \boxed{
        \sum_{ u \in  V^{2ai}} |N_u^+|^p 
        = \sum_{ u \in  V^{2ai}} \Delta_u^p 
        \leq \sum_{ u \in  V^{2ai}} |N_u^- \cap C(u)|^p 
        \leq \sum_{ u \in  V^{2ai}} y(u)^p 
        \leq \textsf{OPT}^p.}
        \end{align*}
        \item \textbf{Case 2a(ii):}\emph{ At least half of $R_1(u)$ is in clusters other than $C(u)$ and $|N_u^- \cap C(u)| \leq \Delta_u$.}\\ 

        Let $ u \in  V^{2a(ii)}$ be the vertices in this case.
    Denote the vertices in $R_1(u)$ that are in clusters other than $C(u)$ by $R_1'(u)$. By definition of Case 2a(ii), $|R_1'(u)| \geq \frac{5}{3} \cdot \Delta_u$. 
        A key fact we will use is that $|C(u)| \leq 2 \cdot \Delta_u$:
        $$|C(u)| = |N_u^- \cap C(u)| + |N_u^+ \cap C(u)| \leq \Delta_u + \Delta_u = 2 \cdot \Delta_u.$$
        For $u \in  V^{2a(ii)}$ and $w \in N_u^+ \cap C(u)$, define 
        $\varphi(u,w) = |R_1'(u) \cap N_w^+|.$

        Each $w \in N_u^+ \cap C(u)$ dispenses 
        $\varphi(u,w)^p/|C(u)|$ charge to $u$. Also, observe that for $v \in R_1'(u)$, we have that $d_{uv} \leq 0.7$, so we know by Proposition \ref{prop: pos_overlap} that $|N_u^+ \cap N_v^+ \cap C(u)| \geq 0.15 \cdot |N_u^+|$. This implies that 
                \begin{align*}\sum_{w \in N_u^+ \cap C(u)} |R_1'(u)  \cap N_w^+|  
                &= \sum_{w \in N_u^+ \cap C(u)} \sum_{v \in R_1'(u)  \cap N_w^+} 1  
                = \sum_{v \in R_1'(u)} \sum_{\underset{ \cap C(u) \cap N_u^+}{w \in  \cap N_v^+}} 1 \\
               & = \sum_{v \in R_1'(u)} | C(u) \cap N_u^+ \cap N_v^+| \geq 
                \sum_{v \in R_1'(u)} 0.15\cdot | N_u^+ | \\
                &= 
                 0.15 \cdot |N_u^+| \cdot |R_1'(u)| \geq 0.15 \cdot \Delta_u \cdot \frac{5}{3}\Delta_u = 0.25 \cdot \Delta_u^2.
                \end{align*}

        First we lower bound the amount of charge each $u$ satisfying Case 2a(ii) receives, which is at least
        \begin{align*}
         \frac{1}{|C(u)|} \sum_{w \in N_u^+ \cap C(u)} \varphi(u,w)^p &\geq \frac{1}{|C(u)|} \cdot \frac{1}{|N_u^+ \cap C(u)|^{p-1}} \cdot \Big(\sum_{w \in N_u^+ \cap C(u)} \varphi(u,w) \Big)^p \\
         &\geq \frac{1}{2\Delta_u} \cdot \frac{1}{\Delta_u^{p-1}} \cdot \Big(0.25\cdot \Delta_u^2 \Big)^p 
         \geq \frac{1}{2} \cdot 0.25^p \cdot |N_u^+|^{p},
         \end{align*}
        where in the first inequality we have applied Jensen's inequality.

    Next we need to upper bound the amount of charge dispensed in total to \textit{all} $u$ satisfying Case 2a(ii). Note by definition that $\varphi(u,w) \leq y(w)$. Each vertex $w \in V$ dispenses at most
    $y(w)^p/|C(u)| = y(w)^p/|C(w)|$
    charge to each $u \in C(w) \cap N_w^+$. So in total each $w$ dispenses at most 
    $|C(w)| \cdot y(w)^p/ |C(w)|  = y(w)^p$
    charge to all $u$ satisfying Case 2a(ii). 

    Now we put together the lower and upper bounds on the total charge dispensed:  
    \begin{align*}
        \sum_{w \in V} y(w)^p &\geq \text{charge dispensed} 
        \geq \sum_{u \in  V^{2a(ii)}} \frac{1}{|C(u)|} \sum_{w \in N_u^+ \cap C(u)} \varphi(u,w)^p 
        \geq \sum_{u \in  V^{2a(ii)}} \frac{1}{2} \cdot 0.25^p \cdot |N_u^+|^p.
    \end{align*}
    \[\text{  In all, } \quad \boxed{\sum_{u \in  V^{2a(ii)}} |N_u^+|^p \leq 2 \cdot 4^p \cdot \sum_{w \in V} y(w)^p \leq 2 \cdot 4^p \cdot \textsf{OPT}^p.}\]

    \end{itemize}

\noindent\textbf{Case 2b: }\emph{At least half of $R_1(u)$ is in $C(u)$.} \\
Let $ u \in  V^{2b}$ be the vertices in this case.
    Denote the vertices in $R_1(u)$ that are in $C(u)$ by $R_1''(u)$. By definition of Case 2b, $|R_1''(u)| \geq \frac{5}{3} \cdot \Delta_u$. Since every vertex in $R''(u)$ is in $N_u^-$, there are at least $|R''(u)|$ disagreements incident to $u$. So
    $y(u) \geq |R''(u)| \geq \frac{5}{3} \cdot \Delta_u,$
     which gives that
    \[ \boxed{\sum_{u \in V^{2b}} |N_u^+|^p = \sum_{u \in V^{2b}} \Delta_u^p \leq \sum_{u \in V^{2b}} y(u)^p \leq \textsf{OPT}^p.}\]
Adding the terms in the boxed expressions across all cases, the proposition follows. 
\end{proof} 
$$
\text{So we have} \quad 
\boxed{S_{11}^+ \leq \sum_{u \in R_1} |N_u^+|^p \leq \left((20/3)^p + 2 + 2 \cdot 4^p \right) \cdot \textsf{OPT}^p.}
$$

Adding together all the cases, we conclude that
$$
    (S^+)^p \leq 2^p \cdot (S_1^+ + S_2^+) 
    \leq 2^p \cdot (S_{11}^+ S_{12^+} + S_2^+) 
    \leq 2^p \cdot [(8^p/2 +1)((20/3)^p + 2 + 2 \cdot 4^p) + 8^p +1 ] \cdot \textsf{OPT}^p.
$$
\end{proof}
\subsubsection{Fractional cost of negative edges in $\ell_p$-norms}
This section bounds the cost of negative edges. The meanings of $\mathcal{C}$, $C(\cdot)$, and $y$ are as in the previous subsection. 
\begin{lemma} \label{lem: lp_bdd-neg}
    For $p \in \mathbb{R}_{\geq 1}$, the fractional cost of the adjusted correlation metric $f$ in the $\ell_p$-norm objective for the set of negative edges is a constant factor approximation to the optimal, i.e.,
   $$(S^-)^p = \sum_{u \in V} \Big(\sum_{v \in N_u^-} (1-f_{uv}) \Big)^p \leq 
2^p ((200/9)^p+1+(10/3)^p+2\cdot (20/3)^p)\cdot \textsf{OPT}^p.
$$
\end{lemma}
\begin{proof}
We have 
\begin{align*}
    (S^-)^p &= \sum_{u \in V} \Big(\sum_{v \in N_u^-} (1-f_{uv}) \Big)^p 
    \leq 2^p \underbrace{\sum_{u \in V} \Big(\sum_{v \in N_u^- \cap C(u)} (1-f_{uv}) \Big)^p}_{S_1^-} + 2^p \underbrace{\sum_{u \in V} \Big(\sum_{v \in N_u^- \cap \overline{C(u)}} (1-f_{uv}) \Big)^p}_{S_2^-}.
\end{align*}

It is easy to bound $S_1^-$ by using the trivial upper bound $1-f_{uv} \leq 1$:

\[\boxed{S_1^- = \sum_{u \in V} \Big(\sum_{v \in N_u^- \cap C(u)} (1-f_{uv}) \Big)^p \leq \sum_{u \in V} \Big(\sum_{v \in N_u^- \cap C(u)} 1 \Big)^p \leq \sum_{u \in V} y(u)^p = \textsf{OPT}^p,}\]
where we have used that every edge $(u,v) \in E^-$ with $v \in C(u)$ is a disagreement incident to $u$. Next, we bound $S_2^-$. Let $R_1$ and $R_2$ be as in the previous subsection: 
$R_1 = \{u : |N_u^- \cap \{v: d_{uv} \leq 0.7\}| \geq \frac{10}{3} \cdot \Delta_u\} $
and $R_2 = V \setminus R_1$. For $u \in R_2$, define 
$$V_u = \{v: v \in N_u^- \cap \overline{C(u)}, d_{uv} \leq 0.7\}.$$
Note that the definition of $V_u$ is the same as $R'_1(u)$ in the previous subsection, 
but here $V_u$ is only defined for $u \in R_2$, while $R'_1(u)$ was defined for $u \in R_1$.
For $u \in R_1$, we have $1-f_{uv} = 0$ for every $v \in V \setminus \{u\}$. So the outer sum in $S_2^-$ only need be taken over $u \in R_2$:
\begin{align*}
    S_2^- &= \sum_{u \in R_2} \Big(\sum_{v \in N_u^- \cap \overline{C(u)}} (1-f_{uv}) \Big)^p 
    \leq \sum_{u \in R_2} \Big(\sum_{\substack{v: v \in N_u^- \cap \overline{C(u)},\\ d_{uv} \leq 0.7} } (1-d_{uv}) \Big)^p 
    \leq \sum_{u \in R_2} |V_u|^p
\end{align*}
In the second equality, we have used that if $u \in R_2$ and $v \in N_u^-$, then $f_{uv} = d_{uv}$, unless $f_{uv}$ was rounded up to 1 in Step 2 of Definition \ref{def: adj_corr_metric} (which happens when $d_{uv} > 0.7$), or $f_{uv}$ was rounded up to 1 in Step 3 (in which case $1-f_{uv} = 0 \leq 1-{d_{uv}}$).

A key observation is that since $u \in R_2$, it is the case that $|V_u| \leq \frac{10}{3} \cdot \Delta_u$. 

\medskip 

Fix a vertex $u \in R_2$. We consider a few cases. 
\paragraph{Case 1: }\emph{At least a 0.15 fraction of $N_u^+$ is in clusters other than $C(u)$.}\\
Define $V^1$ to be the set of $u \in R_2$ that satisfy Case 1. Then for $u \in V^1$,  $0.15 \cdot |N_u^+| \leq y(u)$, and further
$$|V_u| \leq \frac{10}{3}\Delta_u \leq \frac{1}{0.15}\cdot \frac{10}{3} y(u) = \frac{200}{9}y(u).$$
\text{It follows that } 
$$ \boxed{\sum_{u \in V^1} |V_u|^p \leq (200/9)^p \cdot \sum_{u \in V^1} y(u)^p \leq (200/9)^p \cdot \textsf{OPT}^p.}$$

\paragraph{Case 2: }\emph{At least a 0.85 fraction of $N_u^+$ is in $C(u)$. }\\ 
Define $V^2$ to be the set of $u \in R_2$ that satisfy Case 2. 
Fix $u \in V^2$ and $v \in V_u$. Define $N_{u,v} = N_u^+ \cap N_v^+ \cap C(u)$. Since $d_{uv} \leq 0.7$ and by the assumption of this case, using Proposition \ref{prop: pos_overlap} we have  
$$|N_{u,v}| = |N_u^+ \cap N_v^+ \cap C(u)| \geq 0.15 \cdot \Delta_u.$$
Observe that since $v \not 
\in C(u)$ for $v \in V_u$, $(v,w)$ is a (positive) disagreement for all $w \in N_{u,v}$. 

\noindent  \textbf{Case 2a: }\emph{$|N_u^- \cap C(u)| \geq \Delta_u$.} 

    Define $V^{2a}$ to be the set of $u \in V^2$ that satisfy Case 2a. Since all edges $(u,v)$ with $v \in N_u^- \cap C(u)$ are disagreements, we have $y(u) \geq \Delta_u$. Recalling that $|V_u| \leq \frac{10}{3} \cdot \Delta_u$ for $u \in R_2$, we have
    \[\boxed{\sum_{u \in V^{2a}} |V_u|^p \leq \sum_{u \in V^{2a}} (10/3\cdot\Delta_u)^p \leq (10/3)^p \cdot \sum_{u \in V^{2a}} y(u)^p \leq (10/3)^p \cdot \textsf{OPT}^p.}\]

\noindent \textbf{Case 2b: }\emph{$|C(u)| \leq 2\Delta_u$.} 

    Define $V^{2b}$ to be the $u \in V^2$ satisfying Case 2b. 
    Fix $w \in N_u^+ \cap C(u)$ and $u \in V^{2b}$. Define 
    $$\varphi(u,w) = |V_u \cap N_w^+|,$$
    i.e. $\varphi(u,w)$ is the number of $v \in V_u$ with $w \in N_{u,v}$. 
    Each $w \in N_u^+ \cap C(u)$ dispenses 
    $\frac{\varphi(u,w)^p}{|C(u)|}$
    charge to $u$. Also, 
    \begin{align*}\sum_{w \in N_u^+ \cap C(u)} \varphi(u,w) &= 
    \sum_{w \in N_u^+ \cap C(u)} |V_u \cap N_w^+|
    = \sum_{w \in N_u^+ \cap C(u)} \sum_{v \in V_u\cap N_w^+}1 \\
    &= \sum_{v \in V_u} \sum_{w \in N_{u,v}}1 
    =\sum_{v \in V_u}  |N_{u,v}| \geq |V_u| \cdot 0.15 \cdot \Delta_u.
    \end{align*}
    
    Now we can lower bound the amount of charge each $u$ satisfying Case 2b receives, which is at least
    \begin{align*}
    \frac{1}{|C(u)|} \sum_{w \in N_u^+ \cap C(u)} \varphi(u,w)^p &\geq \frac{1}{|C(u)|} \cdot \frac{1}{|N_u^+ \cap C(u)|^{p-1}} \Big(\sum_{w \in N_u^+ \cap C(u)} \varphi(u,w) \Big)^p \\
    &\geq \frac{1}{2\Delta_u} \cdot \frac{1}{\Delta_u^{p-1}} \left(|V_u| \cdot 0.15\cdot \Delta_u \right)^p 
    = \frac{1}{2} \cdot 0.15^p \cdot |V_u|^p,
    \end{align*}
    where in the first line we used Jensen's inequality. 

    To upper bound the amount of charge dispensed in total to \textit{all} $u$ satisfying Case 2b, first note that $\varphi(u,w) \leq y(w)$. Also, each vertex $w \in V$ only distributes charge to $u \in C(w) \cap N_w^+$, and the amount of charge distributed to each such $u$ is 
    $$\frac{\varphi(u,w)^p}{|C(u)|} = \frac{\varphi(u,w)^p}{|C(w)|} \leq \frac{y(w)^p}{|C(w)|},$$
    so that in total each $w$ dispenses at most 
    $\frac{y(w)^p}{|C(w)|} \cdot |C(w)| \leq y(w)^p$ 
 charge. Putting together the lower and upper bounds on the amount of charge dispensed:
 \begin{align*}
     \sum_{w \in V} y(w)^p &\geq \text{total charge dispensed} 
     \geq \sum_{u \in V^{2b}} \frac{1}{|C(u)|} \sum_{w \in N_u^+ \cap C(u)} \varphi(u,w)^p 
     \geq \sum_{u \in V^{2b}} \frac{1}{2} \cdot 0.15^p \cdot |V_u|^p.
 \end{align*}
 So in all,
 \[\boxed{\sum_{u \in V^{2b}} |V_u|^p \leq 2 \cdot (20/3)^p \cdot \sum_{w \in V} y(w)^p = 2 \cdot (20/3)^p \cdot \textsf{OPT}^p.}\]
Adding together all the cases, we see that
$ (S^-)^p\leq 
2^p ((200/9)^p+1+(10/3)^p+2\cdot (20/3)^p)\cdot \textsf{OPT}^p.
$
 \end{proof}

 \subsection{Proofs of Theorem \ref{thm: main-acm} and Corollary {\ref{cor: main-acm-sparse}}} \label{sec: thm1_proof}

Here we show that Theorem \ref{thm: main-acm} and Corollary \ref{cor: main-acm-sparse} follow directly from the preceding lemmas.

\begin{proof}[Proof of Theorem \ref{thm: main-acm}]
First we show that Lemma \ref{lem:lp-bdd} implies that the clustering resulting from inputting $f$ into the KMZ rounding algorithm is $O(1)$-approximate in any $\ell_p$-norm. Since the rounding algorithm does not depend on $p$, the clustering will be the same for all $p$. 
Let $\mathcal{C^*}$ be the clustering produced by running the KMZ rounding algorithm with the adjusted correlation metric $f$ as input. Let $\textsf{ALG}(u)$ be the number of edges incident to $u$ that are disagreements with respect to $\mathcal{C}^*$. From \cite{KMZ19} and Lemmas \ref{lem: adjusted_tri_ineq} and \ref{lem: apx-tri}, we have that for every $u \in V$,
$\textsf{ALG}(u) \leq 12 \cdot y_u$
where $y_u$ is as in LP \ref{KMZ_LP} when taking $x=f$. 
So $||y||_p$ is the fractional cost of $f$ in the $\ell_p$-norm. 
In what follows, $||\textsf{ALG}||_p$ is the objective value of $\mathcal{C}^*$ in the $\ell_p$-norm and $\textsf{OPT}(p)$ is the optimal objective value in the $\ell_p$-norm. 
Thus using Lemma \ref{lem:lp-bdd} in the last inequality, we have 
\[||\textsf{ALG}||_p \leq 12 \cdot ||y||_p \leq 12 \cdot 529 \cdot \textsf{OPT}(p) = 6348 \cdot \textsf{OPT}(p).\]

We show the overall run-time is $O(n^\omega)$. Recall from the analysis in \cite{DMN23} that computing the correlation metric takes time $O(n^\omega)$, and the KMZ rounding algorithm takes time $O(n^2)$. We just have to show the post-processing of $d$ in Steps 2 and 3 of Definition \ref{def: adj_corr_metric} that were done in order to obtain the adjusted correlation metric $f$ can be done quickly. Indeed, Step 2 takes $O(n^2)$ time as it simply iterates through the edges. Step 3 also takes $O(n^2)$ time, since it visits each vertex and iterates through the neighbors. Thus, the run-time remains $O(n^\omega)$. 
\end{proof}

\begin{proof}[Proof of Corollary \ref{cor: main-acm-sparse}]
    In \cite{DMN23}, they observe that to reduce the run-time from $O(n^\omega)$ to $O(n\Delta^2 \log n)$ for graphs with maximum positive degree bounded by $\Delta$, one can compute the correlation metric $d$ in $O(n\Delta^2)$ time and then run the KMZ rounding algorithm in $O(n\Delta^2 \log n)$ time (see the proof of Corollary 1.2 in Appendix D of \cite{DMN23}). We need only compute $f_{uv}$ when $f_{uv}<1$. Otherwise, we can handle $f_{uv} = 1$ implicitly. As in \cite{DMN23}, for each $u \in V$, we can maintain a list of  $d_{uv}$ for all $v$ with $d_{uv} < 1$. Computing these lists takes $O(n\Delta^2)$ time in total. (This is because there are at most $\Delta^2$ vertices $v$ that are distance two away from $u$ in the positive subgraph.) Steps 2 and 3 only prune these lists further, since some distances that are below 1 may be raised to 1; importantly, no distance that is already equal to 1 under $d$ will be reduced in Steps 2 and 3. To carry out Step 2, for each vertex $u$ it takes $O(\Delta^2)$ time to iterate through the list for $u$ and raise the appropriate $d_{uv}$ to 1. Similarly, to carry out Step 3, for each vertex $u$ it takes $O(\Delta^2)$ time to determine whether the condition in Step 3 is satisfied; if it is, we just handle the vertex $u$ implicitly, as all distances $f_{uv}$ are set to 1. 

    Since each vertex's list still has size at most $\Delta^2$ after the post-processing in Steps 2 and 3, the KMZ rounding algorithm with $f$ as input takes time $O(n\Delta^2 \log n)$, by the exact same argument as in \cite{DMN23}.  
\end{proof}

\section{Conclusion}

This paper considered correlation clustering on unweighted, complete graphs, 
a problem that arises in many settings including community detection and the study of large networks.
All previous works that study minimizing the $\ell_p$-norm (for $p \in \mathbb{R}_{>1}$) of the disagreement vector rely on solving a large, convex relaxation (which is costly to the algorithm's run-time) and produce a solution that is only $O(1)$-approximate for one specific value of $p$. 
We innovate upon this rich line of work by 
(1) giving the first combinatorial algorithm for the $\ell_p$-norms for $p \in \mathbb{R}_{>1}$, 
(2) designing scalable algorithms for this practical problem, and 
(3) obtaining solutions that are $O(1)$-approximate for all $\ell_p$-norms (for $p \in \mathbb{R}_{\geq1} \cup \{\infty\}$) simultaneously.  
We note this last point is particularly important, 
as such solutions are good in both global and local senses, 
and thus may be more desirable than typical optimal or approximate solutions for correlation clustering. The existence of these solutions reveals a surprising structural property of correlation clustering.

One question is whether there is a simpler existential (not necessarily algorithmic) proof that there exists an $O(1)$-approximation for the all-norm objective for correlation clustering.

It is also of interest to implement the KMZ algorithm with the adjusted correlation metric as input, and empirically gain an understanding of how good the adjusted correlation metric is for different $\ell_p$-norms. We suspect that our analysis is lossy (for instance, we did not attempt to optimize constants), and that the approximation obtained would be of much better quality than our analysis guarantees.

Finally, it would be interesting if an analogous result can be obtained for weighted correlation clustering.  

\newpage 

\printbibliography

\newpage 

\appendix

\section{Triangle Inequality} \label{sec: tri_ineq}

Below is the proof that the adjusted correlation metric $f$ satisfies an approximate triangle inequality. 

\begin{proof}[Proof of Lemma \ref{lem: adjusted_tri_ineq}]
    Note that after Step 1 of Definition \ref{def: adj_corr_metric}, $f$ satisfies the triangle inequality, since $f=d$ here and for any $u,v,w$, $d_{uv} \leq d_{uw}+d_{vw}.$
    We then consider the execution of Step 2 in Definition \ref{def: adj_corr_metric}.
    The triangle inequality on the adjusted correlation metric holds if either of $(u,w)$ or $(w,v)$ have distance 1 with respect to the adjusted correlation metric, and it also clearly holds if the three edges have the same distance as in the correlation metric. 
    It follows that the only case to check is when $f_{uv}=1$, $f_{uw}=d_{uw}$, and $f_{wv}=d_{wv}$.
    If $f_{uv}=1$ at the end of Step 2, by the definition of $f$ it must be that $d_{uv} > 0.7$.
    Since $d_{uv} \leq d_{uw}+d_{vw}$, we have that $0.7 \leq d_{uw}+d_{vw}$, 
    and so $f_{uv}=1 \leq \frac{10}{7}(d_{uw}+d_{vw}) \leq \frac{10}{7}(f_{uw}+f_{vw})$ at the end of Step 2.
    
    Finally, Step 3 of Definition \ref{def: adj_corr_metric} also preserves this approximate triangle inequality. 
    The $f$ values on the edges leaving a vertex are raised to 1 together, and any triangle involving one of these edges necessarily involves another.
\end{proof}

\section{$\ell_{\infty}$-norm Approximation} \label{sec: infty_norm}
Below we include the proof that the adjusted correlation metric $f$ has bounded fractional cost in the $\ell_{\infty}$-norm. 

\begin{proof}[Proof of Lemma \ref{lem: infty_proof}]
    The fractional cost of the correlation metric $d$ is at most $8 \cdot \textsf{OPT}$ by \cite{DMN23}. That is,
    \begin{equation} \label{eq: frac_cost_infty}
        \sum_{v \in N_u^+} d_{uv} + \sum_{v \in N_u^-} (1-d_{uv}) \leq 8 \cdot \textsf{OPT} \hspace{0.3cm} \text{for all } u \in V.
    \end{equation}

    Consider any $u \in V$. Observe that the second sum in (\ref{eq: frac_cost_infty}) can only decrease if Step 2 in Definition \ref{def: adj_corr_metric} is executed for any of the negative edges incident to $u$. Now let $R_1$ denote the set of vertices $u$ such that Step 3 is executed on $u$. For $u \in R_1$, the  second sum decreases by at least $\frac{3}{10} \cdot \frac{10}{3} \cdot \Delta_u = \Delta_u $, since there are at least $\frac{10}{3}\cdot \Delta_u$ negative edges $(u,v)$ with $1-d_{uv} \geq \frac{3}{10}$, and these edges are all rounded up  to 1 in $f$, so $1-d_{uv}$ is rounded down to 0. On the other hand, the first sum increases by at most $\Delta_u$, since there are $\Delta_u$ positive edges incident to $u$ and each gets rounded up to 1 in $f$. Therefore, we know that (\ref{eq: frac_cost_infty}) does not increase when we replace $d$ with $f$ when $u \in R_1$. 

Finally, it is possible that (\ref{eq: frac_cost_infty}) increases even when $u \not \in R_1$. This is because $u$ may have neighbors $v$ such that $v \in R_1$. For such $v$, $d_{uv}$ is raised to $f_{uv}=1$. This is only better for the second sum in (\ref{eq: frac_cost_infty}). Now we consider what happens to the first sum in (\ref{eq: frac_cost_infty}) when we replace $d$ with $f$. Let $R_2 = V \setminus R_1$. Fix $u \in R_2$. Then
\begin{align*}
    \sum_{v \in N_u^+} f_{uv} &= \sum_{v \in N_u^+ \cap R_1} 1 + \sum_{v \in N_u^+ \cap R_2} d_{uv} \\
    &= \sum_{\substack{v \in N_u^+ \cap R_1 \\ d_{uv} \leq 1/4}} 1 + \sum_{\substack{v \in N_u^+ \cap R_1 \\ d_{uv} > 1/4}} 1 + \sum_{v \in N_u^+ \cap R_2} d_{uv} \\
    &\leq \sum_{\substack{v \in N_u^+ \cap R_1 \\ d_{uv} \leq 1/4}} 1 + \sum_{\substack{v \in N_u^+ \cap R_1 \\ d_{uv} > 1/4}} 4 \cdot d_{uv} + \sum_{v \in N_u^+ \cap R_2} d_{uv} \\
    &\leq \sum_{\substack{v \in N_u^+ \cap R_1 \\ d_{uv} \leq 1/4}} 1 + 4 \cdot \sum_{v \in N_u^+} d_{uv}
\end{align*}
The second term above is bounded by $32 \cdot \textsf{OPT}$ by (\ref{eq: frac_cost_infty}). To bound the first term, let $w^*$ be an arbitrary vertex in $N_u^+ \cap R_1$ (if this set is empty, we just ignore the first term). Since $d_{uw^*} \leq 1/4$, we know by Proposition \ref{prop: similar_nbhds} that $\Delta_u \leq \frac{7}{3} \cdot \Delta_{w^*}$. Also, 
$$ \sum_{v \in N_{w^*}^+} f_{w^*v} + \sum_{v \in N_{w^*}^-} (1-f_{w^*v})  \leq \Delta_w^*$$ 
since $w^* \in R_1$, and by the argument in the previous paragraph we have that $\Delta_{w^*} \leq 8 \cdot \textsf{OPT}$. So
\[\sum_{\substack{v \in N_u^+ \cap R_1 \\ d_{uv} \leq 1/4}} 1 \leq \Delta_u \leq \frac{7}{3} \cdot \Delta_{w^*} \leq \frac{7}{3} \cdot 8 \cdot \textsf{OPT}.\] This concludes the proof. 
\end{proof}

\section{$\ell_1$-norm Approximation}\label{sec:l1-norm}
The next lemma proves that the adjusted correlation metric well approximates the optimal in the $\ell_1$-norm.
We will show that the fractional cost of the adjusted correlation metric $f$ in the $\ell_1$-norm is a constant factor approximation to the optimal, i.e., 
    $$\sum_{u \in V} \sum_{v \in N_u^+} f_{uv} + \sum_{u \in V} \sum_{v \in N_u^-} (1-f_{uv}) = O(\textsf{OPT}).$$
We argue the second sum---over the negative edges---has a constant factor approximation in the first lemma (Lemma \ref{lem: l1_bdd-neg}), then show the first sum---over the positive edges---does too in the second lemma (Lemma \ref{lem: l1_bdd-pos}). 
Combining these lemmas gives the following approximation guarantee for the $\ell_1$-norm.

\begin{lemma}
    The fractional cost of the adjusted correlation metric $f$ in the $\ell_1$-norm objective is a constant factor away from the cost of the optimal integral solution in the $\ell_1$-norm
\end{lemma}
\begin{proof}
The fractional cost in the $\ell_1$-norm is
     $$\sum_{u \in V} \sum_{v \in N_u^+} f_{uv} + \sum_{u \in V} \sum_{v \in N_u^-} (1-f_{uv}) \leq 74 \cdot \textsf{OPT}.$$
The first sum is bounded by $34\cdot \textsf{OPT}$ by Lemma \ref{lem: l1_bdd-pos} and the second sum is bounded by $40\cdot \textsf{OPT}$ by Lemma \ref{lem: l1_bdd-neg}. 
\end{proof}

\begin{lemma} \label{lem: l1_bdd-pos}
    The fractional cost of the adjusted correlation metric $f$ in the $\ell_1$-norm objective for the set of positive edges is a constant factor approximation to the optimal, i.e.,
    $$ \sum_{u \in V} \sum_{v \in N_u^+} f_{uv} = 34 \cdot \textsf{OPT}.$$
\end{lemma}

\begin{proof}[Proof of Lemma \ref{lem: l1_bdd-pos}]

Fix an optimal clustering $\mathcal{C}$. 
We partition vertices based on membership in $C(u)$ or $\overline{C(u)}$ and let $y$ denote the disagreement vector of $\mathcal{C}$. 
We prove in 
the following claim that the total fractional cost of edges $e \in E^+$ such that $f_e = d_e$ is bounded.  
Then, it just remains to bound the fractional cost of positive edges whose distances are raised to 1 in Step 3 of Definition \ref{def: adj_corr_metric}.
\begin{claim}\label{claim: l1-pos-cm}
The total fractional cost of positive edges for the $\ell_1$-norm objective when
$f_e = d_e$ is a $3$-approximation. 
\end{claim}
\begin{proof}[Proof of Claim \ref{claim: l1-pos-cm}]
For the $\ell_1$-norm, $\textsf{OPT} = \sum_{u \in V} y(u)$.
So for a fixed vertex $u$,
we partition the fractional cost of disagreements from positive edges incident to $u$, i.e. $(u,v) \in E^+$, based on whether vertices $v$ are in $C(u)$ or not:
\begin{align*}
    \sum_{u \in V} \sum_{v \in N_u^+} f_{uv}
    &=  \sum_{u \in V}\sum_{v \in N_u^+ \cap C(u)} d_{uv}+ \sum_{u \in V}\sum_{v \in N_u^+ \cap \overline{C(u)}} d_{uv} \\
    &= \sum_{u \in V}\sum_{v \in N_u^+ \cap C(u)} \frac{|N_u^+ \cap N_v^-| + |N_u^- \cap N_v^+|}{n-|N_u^- \cap N_v^-|} + \sum_{u \in V}\sum_{v \in N_u^+ \cap \overline{C(u)}} d_{uv}. \notag 
\end{align*}
The right-hand sum is easy as we can just use that $d_{uv}$ is at most 1. 
$$\boxed{\sum_{u \in V}\sum_{v \in N_u^+ \cap \overline{C(u)}} d_{uv} \leq \sum_{u \in V}\sum_{v \in N_u^+ \cap \overline{C(u)}} 1 \leq \sum_{u \in V} y(u) = \textsf{OPT}}
$$

For the left-hand sum, we will exploit that $d_{uv}$ is symmetric, and specifically that we can bound the denominator by $|N_u^+|$ or $|N_v^+|$.  
\begin{align*}
    &\sum_{u \in V}\sum_{v \in N_u^+ \cap C(u)} \frac{|N_u^+ \cap N_v^-| + |N_u^- \cap N_v^+|}{n-|N_u^- \cap N_v^-|}
    \leq \sum_{u \in V} \sum_{v \in N_u^+ \cap C(u)}\frac{y(u) + y(v)}{n-|N_u^- \cap N_v^-|}\\
    &\leq \sum_{u \in V} \sum_{v \in N_u^+} \frac{y(u) + y(v)}{n-|N_u^- \cap N_v^-|} 
    = \sum_{u \in V} \sum_{v \in N_u^+} \frac{y(u)}{n-|N_u^- \cap N_v^-|} + \sum_{u \in V} \sum_{v \in N_u^+} \frac{y(v)}{n-|N_u^- \cap N_v^-|}\\
    &= \sum_{u \in V} \sum_{v \in N_u^+} \frac{y(u)}{n-|N_u^- \cap N_v^-|} + \sum_{v \in V} \sum_{u \in N_v^+} \frac{y(v)}{n-|N_u^- \cap N_v^-|} \\
    &\leq \sum_{u \in V} \frac{y(u)}{|N_u^+|} \cdot |N_u^+| + \sum_{v \in V} \frac{y(v)}{|N_v^+|} \cdot |N_v^+| \leq 2\cdot \textsf{OPT}.
\end{align*}
All in all:
$$
    \boxed{\sum_{u \in V}\sum_{v \in N_u^+ \cap C(u)} \frac{|N_u^+ \cap N_v^-| + |N_u^- \cap N_v^+|}{n-|N_u^- \cap N_v^-|}
    \leq 2\cdot \textsf{OPT}}
$$
\end{proof}

Now, we must bound the cost of the positive edges whose distances were raised in Step 3 of Definition \ref{def: adj_corr_metric}. 
Let $R_1 = \{u : \vert\{N_u^- \cap \{v: d_{uv} \leq 0.7\}\vert \geq \frac{10}{3}\cdot \Delta_u\}$.
For any vertex $u$, define the shorthand $R_1(u) = \{N_u^- \cap \{v: d_{uv} \leq 0.7\}\}$
\footnote{The use of a subscript may not be clear here, but it is helpful in the $\ell_p$-norm proof.}; 
recall that $u \in R_1$ is put in its own cluster by the adjusted correlation metric.
Then, our next goal is to show that
\begin{equation} \label{eq: ell1_R1}
2 \cdot \sum_{u \in R_1} |N_u^+| = O(\textsf{OPT}).
\end{equation}
Note that the factor of 2 comes from the observation that even if $u \not \in R_1$, $u$ may have neighbors $v \in R_1 \cap N_u^+$, in which cases $f_{uv}$ is raised to 1 even though Step 3 is not executed on $u$. Fix a $u \in R_1$.
We subdivide into several cases based on $\mathcal{C}$.

\paragraph{Case 1: }\emph{At least a 0.15 fraction of $N_u^+$ is in clusters other than $C(u)$.} \\
Let $ u \in V^{1}$ be the vertices in this case.
There are at least $0.15\cdot |N_u^+|$ disagreements incident to $u$, so 
$$\boxed{\sum_{u \in V^1} |N_u^+| \leq \sum_{u \in V^1} \frac{1}{0.15} \cdot y(u) \leq \frac{20}{3}\cdot \textsf{OPT}.}$$

\paragraph{Case 2: }\emph{At least a 0.85 fraction of $N_u^+$ is in $C(u)$.} \\
We subdivide Case 2 into two cases depending on how much 
of $R_1(u)$ is in $\overline{C(u)}$.\\

\noindent \textbf{Case 2a:} \emph{At least half of $R_1(u)$ is in clusters other than $C(u)$.} \\
Let $ u \in  V^{2a}$ be the vertices in this case.
Denote this subset of $R_1(u)$ by $R'_1(u) = R_1(u) \cap \overline{C(u)}$, and note that $|R'_1(u)| \geq \frac{5}{3}\cdot \Delta_u$. If the number of negative edges incident to $u$ that are contained in $C(u)$ is more than $\Delta_u$, then charge $|N_u^+| = \Delta_u$ to $y(u)$, since these negative edges are disagreements. These $u \in V^{2a}$ thus contribute $\textsf{OPT}$ to (\ref{eq: ell1_R1}). 

Otherwise, $|C(u)| \leq 2\Delta_u$.  Consider $v \in R'_1(u)$. Again, let $N_{u,v} = N_u^+ \cap N_v^+ \cap C(u)$. Since $d_{uv} \leq 0.7$ and by the assumption bringing us into Case 2, we apply Proposition \ref{prop: pos_overlap} to see that $|N_{u,v}| \geq 0.15 \cdot |N_u^+|$. 
For each $v \in R'_1(u)$ and $w \in N_{u,v}$, the positive edge $(w,v)$ is a disagreement.

   For $u \in  V^{2a}$ and $w \in N_u^+ \cap C(u)$, define 
    $$\varphi(u,w) = |R_1'(u)\cap N_w^+  |.$$
    Let each $w \in N_u^+ \cap C(u)$ dispense charge $\varphi(u,w)/|C(u)|$ to $u$.
First, observe that $u$ receives $\Omega(|N_u^+|)$ charge; specifically, $u$ receives charge  at least 
 \begin{align*}
    \sum_{w \in N_u^+ \cap C(u)} \frac{\varphi(u,w)}{|C(u)|}
    &= \sum_{w \in N_u^+ \cap C(u)} \frac{|R_1'(u)  \cap N_w^+|}{|C(u)|}  \\
       & = \frac{1}{|C(u)|}\cdot\sum_{w \in N_u^+ \cap C(u)} \sum_{v \in R_1'(u)  \cap N_w^+} 1  
    = \frac{1}{|C(u)|}\cdot\sum_{v \in R_1'(u)} \sum_{\underset{ \cap C(u) \cap N_u^+}     {w \in  \cap N_v^+}} 1 \\
        &= \frac{1}{|C(u)|}\cdot\sum_{v \in R_1'(u)} | C(u) \cap N_u^+ \cap N_v^+| \geq 
        \frac{1}{|C(u)|}\cdot\sum_{v \in R_1'(u)} 0.15\cdot | N_u^+ | \\
    &\geq \frac{1}{|C(u)|}\cdot 0.15 \cdot |N_u^+| \cdot |R_1'(u)| \geq \frac{1}{2\Delta_u}\cdot 0.15 \cdot \Delta_u \cdot \frac{5}{3}\Delta_u = \frac{1}{8} \cdot \Delta_u.
 \end{align*}

Next we upper bound the charge dispensed in total to \emph{all} $u \in V^{2a}$. Note by definition that $\varphi(u,w) \leq y(w)$. Every vertex $w \in V$ dispenses at most
    $$\frac{y(w)}{|C(u)|} = \frac{y(w)}{|C(w)|}$$
    charge to each $u \in C(w) \cap N_w^+$. 
    So each $w$ dispenses a total of charge at most 
    \[\frac{y(w)}{|C(w)|} \cdot |C(w)| = y(w)\]
     to all $u \in V^{2a}$. 
    Putting the lower and upper bounds together:  
    \begin{align*}
        \sum_{w \in V} y(w) &\geq \text{total charge dispensed} \\
        &\geq \sum_{u \in  V^{2a}: |C(u)| \leq 2\Delta_u} \frac{1}{|C(u)|} \sum_{w \in N_u^+ \cap C(u)} \varphi(u,w) \\
        &\geq \sum_{u \in  V^{2a}: |C(u)| \leq 2\Delta_u} \frac{1}{8} \cdot |N_u^+|
    \end{align*}
    so 
    \[\boxed{\sum_{u \in  V^{2a}} |N_u^+| \leq \textsf{OPT} + 8  \cdot \sum_{w \in V} y(w) \leq 9 \cdot \textsf{OPT}.}\]

\noindent \textbf{Case 2b:} \emph{At least half of $R_1(u)$ is in $C(u)$.}\\
Let $ u \in  V^{2b}$ be the vertices in this case.
There are at least $\frac{5}{3}\Delta_u$ negative disagreeing edges incident to $u$ (those edges $(u,w)$ with $w \in R_1(u)$). So we may charge $|N_u^+|$ to $y(u)$:
$$\boxed{\sum_{u \in V^{2b}} |N_u^+| \leq  \textsf{OPT}}$$
This concludes the bounding of the fractional cost of the negative edges, 
and our constant in the approximation factor is $3+ 20/3 + 9 + 1 < 20$ (adding together the boxed upper bounds).
\end{proof}

\begin{lemma} \label{lem: l1_bdd-neg}
    The fractional cost of the adjusted correlation metric $f$ in the $\ell_1$-norm objective for the set of negative edges is a constant factor approximation to the optimal, i.e.,
    $$ \sum_{u \in V} \sum_{v \in N_u^-} (1-f_{uv}) \leq 40\cdot \textsf{OPT}.$$
\end{lemma}

\begin{proof}[Proof of Lemma \ref{lem: l1_bdd-neg}]
First observe that negative edges that have been rounded up to 1 contribute 0 to the fractional cost. So we may focus on the fractional cost of negative edges $e=(u,v)$ with $d_{uv} = f_{uv} \leq 0.7$. \\

Fix an optimal clustering $\mathcal{C}$, whose value is $\textsf{OPT}$.
If $d_{uv} \leq 0.7$ but $u$ and $v$ are in the same cluster in $\mathcal{C}$,
then $(u,v)$ is a disagreement with respect to $\mathcal{C}$,
so we can charge the fractional cost $1-f_e = 1-d_e \leq 1$ to \boxed{$\textsf{OPT}$}. 

Now we need to consider the negative edges whose endpoints are in different clusters in $\mathcal{C}$.
Recall $C(u)$ is the cluster of $\mathcal{C}$ containing vertex $u$. 
We define the set of \emph{cut edges} incident to $u$ to be
$$E_u = \{(u,v) : (u,v) \in E^-, d_{uv} \leq 0.7, C(u) \neq C(v) \}.$$ 

By Step 3 in the definition of the adjusted correlation metric, $|E_u| \leq 10/3\cdot\Delta_u$. Since the fractional cost of each cut edge is $1-f_{uv} = 1-d_{uv} \geq 0.3 = \Omega(1)$, we need to show that the total number of cut edges is bounded by $\textsf{OPT}$, i.e.,
$\sum_{u \in V} |E_u| = O(\textsf{OPT}).$
We let $V_u = \{v: (u,v) \in E_u\} $,
so our goal is equivalent to showing that 
$\sum_{u \in V}|V_u| = O(\textsf{OPT}).$
The definition of $V_u$ is the same as $R'_1(u)$ in the proof of the previous lemma, 
but $R_1'(u)$ was only defined for $u \in R_1$.

We will charge the value of $\sum_{u \in V} |V_u|$ to disagreements in $\mathcal{C}$. We use $y(u)$ to denote the number of disagreements incident to $u$ with respect to $\mathcal{C}$. 

Fix a vertex $u$. We consider several cases based on the clustering $\mathcal{C}$. 

\paragraph{Case 1: }\emph{At least a 0.15 fraction of $N_u^+$ is in clusters other than $C(u)$.}  \\
Let $V^1$ be the vertices of $V$ in this case. 
For any such $u \in V^1$, 
there are at least $0.15\cdot \Delta_u$ disagreements incident to $u$ in $\mathcal{C}$, 
so $0.15\cdot \Delta_u\leq y(u).$ 
Combining this with the fact that $|V_u| \leq \frac{10}{3} \cdot \Delta_u $, 
we see that
$$|V_u| \leq \frac{10}{3}\cdot \Delta_u \leq \frac{200}{9} \cdot y(u).$$
Summing over all $u \in V^1$, we see that
$$\boxed{\sum_{u \in V^1} |V_u| \leq \frac{200}{9} \cdot \sum_{u \in V} y(u) = \frac{200}{9} \cdot \textsf{OPT}.}$$

\paragraph{Case 2: }\emph{At least a 0.85 fraction of $N_u^+$ is in $C(u)$.} \\
Let $V^2$ be the vertices of $V$ in this case. 
For any such $u \in V^2$, we will show that for every $v \in V_u$, 
there is large overlap between $C(u)$, $N_u^+$, and $N_v^+$. 
Proposition \ref{prop: pos_overlap} will help us as we further refine Case 2 into two settings: 
\medskip

\noindent \textbf{Case 2a:} \emph{At least 0.85 fraction of $N_u^+$ is in $C(u)$, and there are at least $\Delta_u$ negative edges incident to $u$ with both endpoints in $C(u)$.}\\ 
Let $V^{2a}$ be the vertices of $V^2$ in this case. 
Since these negative edges are disagreements, we may charge $|V_u|$ to $y(u)$:
$$\boxed{\sum_{u \in V^{2a}} |V_u| \leq \frac{10}{3} \cdot \sum_{u \in V} \Delta_u \leq \frac{10}{3}\cdot \sum_u y(u) = \frac{10}{3} \cdot \textsf{OPT}.}$$

\noindent \textbf{Case 2b:} \emph{At least 0.85 fraction of $N_u^+$ is in $C(u)$, and $|C(u)| \leq 2 \cdot \Delta_u$.}\\
Observe Cases 2a and 2b cover the full range of Case 2. 
This is because we can write $|C(u)| = |C(u)\cap N_u^+| + |C(u)\cap N_u^-|$, 
and if there are at \emph{most} $\Delta_u$ negative edges incident to $u$ with both endpoints in $C(u)$, this is equivalent to saying $|C(u)\cap N_u^-| \leq \Delta_u$.
So, both $|C(u)\cap N_u^+|$ and $|C(u)\cap N_u^-|$ are upper bounded by $\Delta_u$ (since $\Delta_u = |N_u^+|$).

\smallskip

Let $V^{2b}$ be the vertices of $V^2$ in this case. 
Let $N_{u,v} := N_v^+ \cap N_u^+ \cap C(u)$. We know from Proposition \ref{prop: pos_overlap} (for $v \in R_1(u)'$ and using the assumption bringing us into Case 2) that $|N_{u,v}| \geq 0.15 \cdot \Delta_u$.  
Observe that for $v \in V_u \subseteq N_u^- \cap \overline{C(u)}$ 
and $w \in N_{u,v}$, the positive edge $(w,v)$ is a disagreement.

Fix $w \in N_u^+ \cap C(u)$ and $u \in V^{2b}$. Define $\varphi(u,w) = |V_u \cap N_w^+|$. We say that $w$ dispenses $\varphi(u,w)/|C(u)|$ charge to $u$. We would like to show that each $u \in V^{2b}$ receives $\Omega(|V_u|)$ charge, and that each $w$ dispenses at most $y(w)$ charge. By double counting,  
$$\sum_{w \in N_u^+ \cap C(u)} \varphi(u,w) \geq |V_u| \cdot 0.15\Delta_u $$
where we have used that $\varphi(u,w)$ is equal to the number of $v \in V_u$ such that $w \in N_{u,v}$, and we know that $|N_{u,v}|$ is at least $0.15 \Delta_u$. So, the total amount of charge that $u \in V^{2b}$ receives is 
$$\sum_{w \in N_u^+ \cap C(u)} \frac{\varphi(u,w)}{|C(u)|} \geq |V_u| \cdot 0.15 \Delta_u / |C(u)| \geq 3/40 \cdot |V_u|$$
since $|C(u)| \leq 2\Delta_u$ by assumption of the case. 

Now we need to upper bound the amount of charge that each $w \in V$ dispenses in total. Note that $\varphi(u,w) \leq y(w)$, and $w$ only dispenses to vertices $u \in C(w) \cap N_w^+$. Further, $C(u) = C(w)$. So $w$ dispenses at most $|C(w)| \cdot y(w)/|C(u)| = y(w)$ charge in total. Putting everything together: 

\begin{align*}
     \sum_{w \in V} y(w) &\geq \text{total charge dispensed} \\
     &\geq \sum_{u \in V^{2b}}  \sum_{w \in N_u^+ \cap C(u)} \frac{\varphi(u,w)}{|C(u)|} 
     \geq \sum_{u \in V^{2b}} \frac{3}{40} \cdot |V_u|.
 \end{align*}
 So in all,
 \[\boxed{\sum_{u \in V^{2b}} |V_u| \leq \frac{40}{3} \cdot \sum_{w \in V} y(w) = \frac{40}{3} \cdot \textsf{OPT}.}\]
This concludes the bounding of the fractional cost of the negative edges, 
and our constant in the approximation factor is $1 + 200/9 + 10/3 + 40/3 < 40$ (adding together the boxed upper bounds).

\end{proof}

\section{KMZ Rounding Algorithm and Run-time}
\label{sec: lp_rounding_alg}

We construct a fractional solution with our adjusted correlation metric (recall we use this in place of an LP solution), and then
input this solution to a rounding algorithm by Kalhan, Makarychev, and Zhou \citeyearpar{KMZ19}.

First, we set up some notation for the rounding algorithm, 
mainly following the example set by Kalhan, Makarychev, and Zhou.
Let the ball of radius $\rho$, with respect to the distance defined by 
any semi-metric $z$ on $V$, around vertex $u \in V$ 
be denoted $\textsf{Ball}(u,\rho) = \{v \in V \mid z_{uv} \leq \rho\}$.
The algorithm is iterative, and continues to find new clusters until every vertex is clustered.
Let the set of unclustered vertices at time $t$ be denoted by $V_t \subseteq V$.
The cluster center in $V_t$ is the vertex maximizing the quantity
$$
L_t(u)= \sum_{v \in \textsf{Ball}(u,r) \cap V_t} (r-z_{uv})
$$
for $r$ a parameter input to the algorithm.
Note that when $L_t$ is large,
vertices in $\textsf{Ball}(u,r) \cap V_t$ are closely clustered together.\\

\begin{algorithm} \label{KMZ-alg}[Rounding algorithm]
\end{algorithm}

\begin{minipage}{13cm}
\rule{12cm}{0.4pt}\\
\noindent \textbf{Input: }Semi-metric $z$ on $V$.

\noindent \textbf{Output: }Clustering $\mathcal{C}$.
\begin{enumerate}
    \item Let $V_0 = V$, $r = 1/5$, $t=0$.
    \item \textbf{while} ($V_t \neq \emptyset$)
    \begin{itemize}
        \item Find $u^*_t = \arg \max_{u \in V_t} L_t(u) =  \arg \max_{u \in V_t} \sum_{v \in \text{Ball}(u,r) \cap V_t} {r-z_{uv}}$. 
        \item Create ${C_t = \text{Ball}(u^*_t,2r) \cap V_t}$.
        \item Set $V_{t+1} = V_t \setminus C_t$ and $t=t+1$.
    \end{itemize}
    \item Return ${\mathcal{C} = (C_0, \dots, C_{t-1})}$.
\end{enumerate}
    \rule{12cm}{0.4pt}
\end{minipage}\\

We define $\textsf{LP}(u,v)$ to be $(u,v)$'s cost to LP \ref{KMZ_LP} in the objective, i.e., $x_{uv}=\textsf{LP}(u,v)$ for $(u,v) \in E^+$ and 
$x_{uv}=1-\textsf{LP}(u,v)$ if $(u,v) \in E^-$. 
Similarly, we define a cost for each edge with respect to 
Algorithm \ref{KMZ-alg},
$\textsf{ALG}(u,v) = \mathds{1}((u,v) \text { is a disagreement})$.
Then summing over all edges incident to vertex $u$, we obtain the disagreements to each vertex, $\textsf{ALG}(u) = \sum_{v \in V} \textsf{ALG}(u,v)$.

The technical work of Kalhan, Makarychev, and Zhou \citeyearpar{KMZ19} is in showing that 
\begin{equation}\label{eq: KMZ}
    \textsf{ALG}(u) = \sum_{v \in V}\textsf{ALG}(u,v) \leq 
5 \cdot \sum_{v \in V}\textsf{LP}(u,v) = 5 y(u).
\end{equation}

Davies, Moseley, and Newman \citeyearpar{DMN23} discuss the run-time of Algorithm \ref{KMZ-alg} (see Appendix A in their paper).
The rounding algorithm has run-time $O(n^2)$ for general graphs, 
and the full run-time of the KMZ algorithm is dominated by the time to solve the LP.

\end{document}